\documentclass[a4paper, 12pt]{article}
\usepackage[left=1.25in,right=1.25in, top=2.5cm,bottom=2.5cm]{geometry}
\usepackage{amsmath, amsfonts, amssymb, amsthm}
\usepackage{multicol, multirow}
\usepackage{graphicx}
\usepackage{enumerate}
\usepackage[utf8]{inputenc}
\usepackage[inline,shortlabels]{enumitem}
\usepackage{bbm, dsfont}
\usepackage{bm}
\usepackage{authblk}
\usepackage[usenames,dvipsnames]{xcolor}
\usepackage[colorlinks=true,linkcolor=blue, citecolor=blue, urlcolor=blue]{hyperref}
\usepackage{epsfig, epstopdf}
\usepackage{caption}
\usepackage[textfont=footnotesize,justification=centering]{subcaption}
\usepackage{appendix}
\usepackage{afterpage}
\usepackage{float}
\usepackage{comment}
\usepackage{pdflscape}
\usepackage{anyfontsize}
\usepackage{soul}
\usepackage[multiple]{footmisc}
\usepackage[sort]{natbib}
\usepackage{hyperref}
\usepackage{stackrel}
\usepackage[normalem]{ulem}
\usepackage{lineno} 
\usepackage{csquotes}
\usepackage{lipsum}
\usepackage{yfonts}
\usepackage[multiple]{footmisc}

\bibpunct{(}{)}{;}{a}{,}{,}

\newtheorem*{theorem-non}{Theorem}
\numberwithin{equation}{section}
\theoremstyle{definition}
\newtheorem{assumption}{Assumption}[section]
\newtheorem{theorem}[assumption]{Theorem}
\newtheorem{definition}[assumption]{Definition}
\newtheorem{proposition}[assumption]{Proposition}
\newtheorem{lemma}[assumption]{Lemma}
\newtheorem{corollary}[assumption]{Corollary}
\newtheorem{remark}[assumption]{Remark}
\newtheorem{example}[assumption]{Example}

\newcommand{\rBrackets}[1]{\left( #1 \right)}
\newcommand{\sBrackets}[1]{\left[ #1 \right]}

\newcommand{\tin}{t \in \mathcal{T}}

\newcommand{\ecu}{\hat{\bm{u}}}
\newcommand{\cu}{\bm{u}}
\newcommand{\cuh}{\bm{u}_h}

\newcommand{\EV}{\mathbb{E}}
\newcommand{\EVtx}{\mathbb{E}_{t, x}}

\newcommand{\EVtxz}{\mathbb{E}_{t, x, z}}

\newcommand{\inttT}{\int \limits_{t}^{T}}
\newcommand{\inttth}{\int \limits_{t}^{t + h}}
\newcommand{\intthT}{\int \limits_{t + h}^{T}}

\newcommand{\R}{\mathbb{R}}
\newcommand{\N}{\mathbb{N}}

\newcommand{\TSet}{\mathcal{T}}
\newcommand{\TSetNoT}{[0, T)}
\newcommand{\XSet}{\mathcal{X}}
\newcommand{\AMap}{\mathcal{A}}
\newcommand{\ASet}{\bm{\mathcal{A}}}
\newcommand{\tx}{(t,x) \in \mathcal{T} \times \mathcal{X}}
\renewcommand{\footnoterule}{\kern -3pt \hrule width 1\textwidth height 0.1pt \kern 2pt}

\DeclareFontFamily{LYG}{ygoth}{}
\DeclareFontShape{LYG}{ygoth}{m}{n}{<->ygoth}{}

\allowdisplaybreaks

\begin{document}

\title{Equilibrium control theory for Kihlstrom\textendash Mirman preferences in continuous time\footnote{We thank Frank Riedel, Moris Simon Strub, and Jianfeng Zhang for helpful comments and suggestions. L. De Gennaro Aquino acknowledges the support of the University of Copenhagen and the Johannes Kepler University Linz, where the vast majority of this work was done.} \\[.5cm]}

\author{Luca De Gennaro Aquino\footnote{Department of Engineering, Reykjavik University, Iceland. \href{lucaa@ru.is}{lucaa@ru.is}}\quad  Sascha Desmettre\footnote{Institute of Financial Mathematics and Applied Number Theory, Johannes Kepler University Linz, Austria. \href{sascha.desmettre@jku.at}{sascha.desmettre@jku.at}} \\ Yevhen Havrylenko\footnote{Department of Mathematical Sciences, University of Copenhagen, Denmark.  \href{yh@math.dk.ku}{yh@math.ku.dk}.} \textsuperscript{,}\footnote{Institute of Insurance Science, University of Ulm, Germany.  \href{yevhen.havrylenko@uni-ulm.de}{yevhen.havrylenko@uni-ulm.de}.} \quad Mogens Steffensen\footnote{Department of Mathematical Sciences, University of Copenhagen, Denmark. \href{mogens@math.dk.ku}{\mbox{mogens@math.ku.dk}}} }

\maketitle

\begin{abstract}
In intertemporal settings, the multiattribute utility theory of Kihlstrom and Mirman suggests the application of a concave transform of the lifetime utility index. This construction, while allowing time and risk attitudes to be separated, leads to dynamically inconsistent preferences. We address this issue in a game-theoretic sense by formalizing an equilibrium control theory for continuous-time Markov processes. In these terms, we describe the equilibrium strategy and value function as the solution of an extended Hamilton\textendash Jacobi\textendash Bellman system of partial differential equations. We verify that (the solution of) this system is a sufficient condition for an equilibrium and examine some of its novel features. A consumption-investment problem for an agent with CRRA-CES utility showcases our approach.
\end{abstract}


\section{Introduction}

This paper studies non-additively separable utility functions of the form
\begin{equation} \label{eq:KMpreferencesIntro}
	\mathbb{E}_{t}\left[ \Phi \left(\int_{s \geq t} H(s,c(s),t) ds \right) \right], \quad t \geq 0,
\end{equation}
where $\mathbb{E}_{t}$ is the expectation operator conditional on information available at \mbox{time $t$}, $H$ is a discounted utility of consumption $c$, and $\Phi$ is a concave function.  We refer to this class of models as Kihlstrom\textendash Mirman (henceforth, KM) preferences, for the application of a concave transform of the time-additive utility index was originally suggested by \cite{KihlstromMirman1974:JET,KihlstromMirman1981:RES} as the basis for an extension of the Arrow\textendash Pratt measure of risk aversion to the case of many commodities; see also \cite{Kihlstrom2009:JME}.

With the premise that agents evaluate risk over lifecycle consumption (rather than locally), the appeal of this formulation is that it allows for a separation between time and risk attitudes. However, the nonlinear aggregation of immediate and prospective utilities induces dynamic 
inconsistency, a concern that \cite{EpsteinZin1989:Econometrica} and \cite{Weil1990:QJE} already pointed out in the lead-up to the development of the recursive utility framework.\footnote{Epstein\textendash Zin\textendash Weil preferences are well-recognized for being time-consistent and maintaining the disentanglement of elasticity of intertemporal substitution and risk aversion, among other properties.} This is precisely the issue that we wish to address here.

\cite{Kihlstrom2009:JME} made the first attempt in this direction and analyzed a two- and three-period investment-consumption problem and an infinite-horizon consumption-based asset pricing model for an agent with CRRA-CES utility. To deal with time-inconsistency, Kihlstrom followed the game-theoretic approach of \cite{Strotz1956:RES}\footnote{In this classic paper, Strotz envisioned three types of behavior in time-inconsistent scenarios. Besides consistent planners who try to reach an intrapersonal equilibrium, he described the strategy of precommitters and spendthrifts. In turn, precommitters acknowledge that in the future, their preferences might change and decide to follow the plan determined at the beginning date, while spendthrifts (often called na\"ive or myopic agents) are unaware of possible intertemporal conflicts and recalculate the course of action at any time.} by searching for consistent planning. The key messages of his analysis are that, in a Lucas' model with independent identically distributed (i.i.d.) consumption growth, (i) KM preferences can generate a higher equity premium than additively separable preferences if and only if the elasticity of substitution is lower than in the additively separable case, and that (ii) in the absence of a riskless asset, savings decrease (increase) as a function of risk aversion if the elasticity of substitution is higher (less) than one. We return to some of these findings in Section \ref{sec:CRRA-CES_utility_exponential_discounting}.

A more general theory for discrete-time Markovian control problems appeared in \cite{BjoerkKhapkoMurgoci2021:TICT} [Chapter 6.3], where the authors provided an extended Bellman system characterizing an equilibrium strategy and the corresponding equilibrium value function. 

Our contribution is to delineate the continuous-time version of this theory. First, we derive the corresponding extended Hamilton\textendash Jacobi\textendash Bellman (HJB) system of partial differential equations (PDEs). We compare the features of this system with other instances emerging from the literature (part of which are nested within our framework), showing that the additional source of time-inconsistency induced by the non-time-additivity of KM preferences can be embedded into a new state variable.   

Then, we verify that, under regularity assumptions, our derived HJB system provides a sufficient condition for an equilibrium. The proof of the verification theorem relies on a new Feynman-Kac-type result for an auxiliary utility function, which is of interest of itself.  

By way of example, we study a consumption-investment problem with CRRA-CES preferences in a Black\textendash Scholes market. In this case, we show that the value function is (multiplicatively) separable in wealth and time, inheriting the structure of the utility function with a time component determined by the solution of a system of ordinary differential equations (ODEs). Interestingly, the resulting investment strategy reduces to the same constant rate of \cite{Merton1969:RES}, upon the distinction between the parameters identifying intertemporal substitution and risk aversion, and under the assumption of CRRA-CES recursive utility; see, e.g., \cite{KraftSeifriedSteffensen2013:FS}. In addition, the consumption strategy is proportional to wealth, with a time-dependent factor related to the solution of the system of ODEs mentioned above. We find the latter through numerical simulations, which reveal the similarity in lifecycle consumption compared to the benchmark recursive utility specification. 

\paragraph*{Related literature.} Without trying to do justice to an extensive body of work, we list some key contributions on various topics found in this paper. 

Discrete-time analogs of \eqref{eq:KMpreferencesIntro}, or special cases, have been featured in theoretical and applied contexts. 
\cite{BommierLeGrand2014:JRU} demonstrated that, when calibrated to mortality data and bequest motives, these preferences can explain the low participation in the annuity market; \cite{AndersenHarrisonLauRutstrom2018:IER} employed them in an experimental setting to elicit participants’ intertemporal correlation aversion (\cite{Richard1975:MS}, \cite{EpsteinTanny1980:CJE}); \cite{DeJarnetteDillenbergerGottliebOrtoleva2020:Econometrica}, under the label Generalized Expected Discounted Utility models, showed that they can accommodate different attitudes towards time lotteries.

As for the recursive utility framework, \cite{Weil1989:JME}, \cite{EpsteinZin1991:JPE}, \cite{Tallarini2000:JME}, and \cite{BansalYaron2004:JF} to name a few, conducted empirical studies on asset pricing, business cycles, and welfare. \cite{KraftSeifried2010:MFE, KraftSeifried2014:JET} further investigated the continuous-time limit (also known as \emph{stochastic differential utility}) of \cite{DuffieEpstein1992:Econometrica}, while \cite{MarinacciMontrucchio2010;JET}, \cite{HansenScheinkman2012:PNAS}, \cite{BorovivckaStachurski2020:JF}, \cite{Christensen2022:JET}, \cite{MonoyiosMostovyi2024:MF} established conditions for existence and uniqueness of the recursive program. Along these lines, \cite{BloiseVailakis2018:JET}, \cite{Balbus2020:ET}, \cite{RenStachurski2021:Automatica},  and \cite{BloiseLeVanVailakis2024:Econometrica} analyzed the impact of different aggregators on dynamic programming. 

Equilibrium strategies for time-inconsistent problems were initially studied in the context of non-exponential discounting: applications in finance and economics include \cite{Pollak1968:RES}, \cite{PhelpsPollak1968:RES}, \cite{PelegYaari1973:RES}, \cite{Goldman1980:RES}, \cite{Laibson1997:QJE}, \cite{EkelandPirvu2008:MFE}, \cite{EkelandLazrak2010:MFE}, \cite{EkelandMbodjiPirvu2012:SIFIN}. 

The literature on mean-variance portfolio selection provides another important window on time-inconsistency. Within this setting, the game-theoretical approach was considered by \cite{BasakChabakauri2010:RFS}, then generalized by \cite{Czichowsky2013:FS},  \cite{BjorkMurgociZhou2014:MF}, and \cite{BieleckiChenCialenco2021:IJTAF} in several directions. \cite{LiNg2000:MF} and \cite{ZhouLi2000:AMO} pioneered the embedding technique to derive precommitted strategies, whereas \cite{PedersenPeskir2017:MFE} and \cite{ChenZhou2024:MF} focused on na\"ive policies. 

Finally, we mention several additional works that have looked at diverse aspects of time-inconsistency. Among others, \cite{HuangZhou2020:MF}, \cite{HuangWang2021:SICON} and \cite{LiangYuan2023:MF} worked on optimal control-stopping problems; \cite{DesmettreSteffensen2023:MF} discussed the case of an agent maximizing the certainty equivalent of terminal wealth with random risk aversion. On more abstract grounds, \cite{KrygerNordfangSteffensen2020:MMRO} provided the verification theorem for a general class of objective functionals with nonlinearity in the conditional expectation; \cite{HeJiang2021:SICON} and \cite{HuangZhou2021:MOR} explored different notions of equilibrium controls; \cite{Lindensjoe2019:ORL} showed that under regularity assumptions the extended HJB system is both a necessary and sufficient condition for an equilibrium; \cite{LeiPun2024:MF} verified the well-posedness (existence and uniqueness) of the solution of the extended HJB system. All these findings assumed a Markovian framework. In the non-Markovian case, \cite{HernandezPossamai2023:AAP} recently made significant progress under additively separable payoff functionals.

\paragraph*{Structure of the paper.} In \mbox{Section \ref{sec:DT_KM}}, we describe the non-additively separable preference specification of \cite{KihlstromMirman1974:JET} and its relation with alternative models of choice disentangling time and risk attitudes. Within the framework of Markovian stochastic control problems, we also review the equilibrium (game-theoretic) approach by \cite{BjoerkKhapkoMurgoci2021:TICT} to address the dynamic inconsistency of Kihlstrom\textendash Mirman preferences in discrete time. We present our extension of the equilibrium theory to continuous time in \mbox{Section \ref{sec:general_theory}}, with the main result being the verification theorem for an extended HJB system of PDEs. For example, \mbox{Section \ref{sec:CRRA-CES_utility_exponential_discounting}} illustrates a consumption-investment problem with CRRA-CES preferences. \mbox{Section \ref{sec:Conclusions}} concludes and outlooks our work. All proofs and additional results are relegated to the Appendices.

\section{Kihlstrom\textendash Mirman preferences} \label{sec:DT_KM}

Given a time horizon $T \in (0,\infty)$, let $\{t_{0} = 0, t_{1},\dots, t_{N} = T \}$ be a partition of $\mathcal{T} := [0,T]$, for some $N \in \mathbb{N}$,\footnote{Throughout, we use the convention that the set of natural numbers $\mathbb{N}$ does not \mbox{include 0.}} and let $c := \{c_{t_{n}}\}_{n = 0}^{N}$ be a discrete-time sequence of consumptions taking values in an interval $\mathcal{C} \subseteq \mathbb{R}$. 

Applying the general multicommodity analysis of \cite{KihlstromMirman1974:JET, KihlstromMirman1981:RES} to an intertemporal setting, we consider the utility function 
\begin{equation}\label{eq:KMpreferences}
	J_{t_{n}}(c) = \mathbb{E}_{t_{n}}\sBrackets{\Phi\rBrackets{\sum_{k=n}^{N}H\big(t_{k},  c_{t_{k}}, t_{n}\big)}},
\end{equation}
with $H : \mathcal{T} \times \mathcal{C} \times \mathcal{T} \to \mathbb{R}$ being a discounted utility, and $\Phi : \mathbb{R} \to \mathbb{R}$ nonlinear and increasing.  The time-additive case is obtained when $\Phi$ is affine.

As mentioned in the Introduction, KM preferences compromise on additive separability -- and dynamic consistency -- for a separation between time and risk attitudes. This distinction can be seen more clearly by writing $H$ as an exponentially discounted utility: $H\big(t_{k},  c_{t_{k}}, t_{n}\big) = \delta^{t_{k}-t_{n}}u(c_{t_{k}})$, for some utility function $u$. Under this provision, marginal rates of intertemporal substitution are determined in the absence of risk by $u$,  while preferences concerning risk (or, put differently, concerning variations in discounted lifetime utility) are measured by $\Phi \, \circ \, u$.

\begin{example}
	\label{example:CRRA-CES_preferences}
	A common specification of \eqref{eq:KMpreferences} is the CRRA-CES case:
	\begin{equation}
		\label{eq:CRRA-CES_preferences}
		J_{t_{n}}\left(c\right) =  \mathbb{E}_{t_{n}}\sBrackets{\dfrac{1}{1-\alpha}\rBrackets{\sum_{k=n}^{N}\delta^{t_{k} - t_{n}}c_{t_{k}}^{\,\rho}}^{\frac{1-\alpha}{\rho}}}.
	\end{equation}
	Here $\alpha \geq 0$ is the coefficient of relative risk aversion concerning the entire consumption stream, and $(1-\rho)^{-1}$ is the elasticity of intertemporal substitution, with $\rho < 1$.
\end{example}

\begin{remark} \label{footnote:KihlstromMirmanRiskAversion}
	It is important to highlight that, in this context, risk attitudes are gauged through a notion of risk aversion that generalizes the standard (single-argument) Arrow\textendash Pratt measure to the case of multiple attributes.  For instance, considering CES preferences, let $U^{0}_{t_{n}}\left(c\right) := \rBrackets{\sum_{k=n}^{N}\delta^{t_{k}-t_{n}}c_{t_{k}}^{\, \rho} }^{1/\rho}$ denote the \textit{least concave representation} of the bundle $c$.  For a strictly concave function $ v : \mathbb{R} \to \mathbb{R}$, the relative risk aversion of $U_{t_{n}}\left(c\right) :=  v\left( U^{0}_{t_{n}}\left(c\right) \right)$, as defined in \cite{KihlstromMirman1974:JET}, is $$ - \dfrac{\, U^{0}_{t_{n}}\left(c\right) v''\left(U^{0}_{t_{n}}\left(c\right)\right)}{v'\left(U^{0}_{t_{n}}\left(c\right)\right)}.$$
	In light of this definition, choosing a CRRA function $v(x) = x^{1-\alpha} / (1-\alpha)$, it should then become evident that the relative risk aversion of the utility in \eqref{eq:CRRA-CES_preferences} is equal to $\alpha$. We refer the reader to \cite{Kihlstrom2009:JME} for further comments on these aspects and to \cite{Debreu1976:JME} for the background on least concave utilities.
\end{remark}

\paragraph*{Separation of time and risk attitudes: alternative approaches.} Distinguishing time and risk attitudes is not a unique feature of KM preferences. For comparison, we mention two alternative approaches that make this possible, starting from the non-expected recursive utility framework axiomatized by \cite{KrepsPorteus1978:Econometrica}, then extended by \cite{EpsteinZin1989:Econometrica} and \cite{Weil1990:QJE}.\footnote{An overview of utility classes that distinguish risk aversion from intertemporal substitution is also offered in \cite{BommierChassagnonLeGrand2012:JET}.} In that case, the elasticities of substitution between periods and between states are separated via locally aggregating the utility of current consumption and the certainty equivalent of future indirect utilities:
\begin{equation} \label{eq:EZWrecursiveutility}
	J_{t_{n}}(c) = W\left( c_{t_{n}}, \mathcal{M}_{t_{n}}\left(J_{t_{n+1}}(c)\right)  \right), 
\end{equation}
where $W : \mathcal{C} \times \mathbb{R} \to \mathbb{R} $ is a nonlinear temporal aggregator, and $\mathcal{M}_{t_{n}}$ a time-$t_{n}$ certainty equivalent of future continuation value $J_{t_{n+1}}$. Time-additivity is obtained if both the certainty equivalent and aggregator are linear, i.e., $\mathcal{M}_{t_{n}}\left(J_{t_{n+1}}(c)\right) = \mathbb{E}_{t_{n}}\left[J_{t_{n+1}}(c)\right]$, and $W\left( c_{t_{n}}, \mathcal{M}_{t_{n}}\left(J_{t_{n+1}}(c)\right)  \right) =  H(c_{t_{n}}) + \delta^{t_{n+1}-t_{n}} \mathcal{M}_{t_{n}}\left(J_{t_{n+1}}(c)\right)$.

\begin{example}
	For a CES aggregator and CRRA certainty equivalent, \eqref{eq:EZWrecursiveutility} becomes
	\begin{equation*} \label{eq:EZWrecursiveutility_example}
		\begin{split}
			J_{t_{n}}(c) = \Big( c_{t_{n}}^{\,\rho} + \delta^{t_{n+1}-t_{n}}  \left( \mathbb{E}_{t_{n}}\left[ \left(J_{t_{n+1}}(c)\right)^{1-\alpha}\right] \right)^{\frac{\rho}{1-\alpha}} \Big)^{\frac{1}{\rho}}.
		\end{split}
	\end{equation*}
\end{example}

Another formulation that allows for disentangling preferences concerning time and risk was recently proposed by \cite{JensenSteffensen2015:IME} and \cite{FahrenwaldtJensenSteffensen2020:JME}. In a continuous-time setup, the authors specify a utility process in which the temporal aggregation is applied globally on the certainty equivalents of direct utilities. In other words, instead of first forming the certainty equivalent of the indirect utility and then nonlinearly time-aggregating with present consumption, they build up a time-global function of utilities of certainty equivalents concerning future uncertain consumptions.\footnote{A similar representation involving ordinal certainty equivalents previously appeared in \cite{Selden1978:Econometrica}.}

We can imagine the discrete-time analog of their construction as follows:
\begin{equation} \label{eq:FahrenwaldtEtAl_TimeGlobalPreferences}
	\begin{split}
		J_{t_{n}}(c) = \psi \Bigg(  \sum_{k=n}^{N} \delta^{t_{k}-t_{n}}\varphi \left( U^{-1} \Big( \mathbb{E}_{t_{n}}\big[  U\left( c_{t_{k}}\right) \big]  \Big) \right) \Bigg),
	\end{split}
\end{equation}
where $U : \mathcal{C} \to \mathbb{R}$ is a von Neumann\textendash Morgenstern utility function, $\varphi: \mathbb{R} \to \mathbb{R}$ represents the agent's preferences concerning local certainty equivalents, and $\psi: \mathbb{R} \to \mathbb{R}$ is an increasing time-global transformation. Here, we recover time-additivity by choosing $\varphi = U$ and $\psi$ as the identity function.

\begin{example} Using again constant EIS, a CRRA utility function, and -- for tractability reasons -- a convenient power function for $\psi$, \eqref{eq:FahrenwaldtEtAl_TimeGlobalPreferences} becomes
	\begin{equation*} \label{eq:FahrenwaldtEtAl_TimeGlobalPreferences_example}
		\begin{split}
			J_{t_{n}}(c) = \dfrac{1}{\rho}\left( \sum_{k=n}^{N} \delta^{t_{k}-t_{n}} \left(\mathbb{E}_{t_{n}}\left[ c_{t_{k}}^{\,\rho} \right] \right)^{\frac{1-\alpha}{\rho}} \right)^{\frac{\rho}{1-\alpha}}.
		\end{split}
	\end{equation*}
\end{example}

While the behavioral motivation behind the above paradigms is, in some measure, shared, only recursive utility preserves the property of time-consistency. In the other cases, the utility's non-time-separable structure leads to dynamically inconsistent decisions.

Fahrenwaldt, Jensen, and Steffensen recast the problem from a game-theoretical perspective by looking for an intra-personal equilibrium strategy to address this issue.\footnote{In \cite{JensenSteffensen2015:IME} this is done regarding a consumption-investment-insurance problem in a classical Black\textendash Scholes market, while \cite{FahrenwaldtJensenSteffensen2020:JME} tackle a consumption-investment problem in a general diffusive incomplete market.} Their approach, which follows in the footsteps of \cite{BjorkKhapkoMurgoci2017:FS}, leads to an extended Hamilton\textendash Jacobi\textendash Bellman equation for the value function of the problem.

Concerning KM preferences, an earlier application of the consistent planning approach was put forward by \cite{Kihlstrom2009:JME}. In contrast, a systematic discrete-time equilibrium control theory for Markovian problems can be found in Chapter 6.3 of \cite{BjoerkKhapkoMurgoci2021:TICT}. We condense this theory below, focusing primarily on the extended Bellman system of equations for the equilibrium value function. This will set the stage for our continuous-time extension in the following sections.

\subsection{Equilibrium theory in discrete time} \label{subsec:EqTheory_kM_DT}
Let $X^{\cu} = \{X_{t_{n}}^{\cu}\}_{n=0}^{N}$ be a controlled Markov process evolving in a state space $\mathcal{X}$, and let $\cu = \{\cu_{t_{n}} \}_{n=0}^{N}$ be a control process taking values in a control space $\bm{\mathcal{A}}$. Formalities are not essential at this point -- further details will be provided in \mbox{Section \ref{sec:general_theory}.}

Recall the Kihlstrom\textendash Mirman utility\\
\begin{equation}\label{eq:KM_DT_reward_func}
	\begin{split}
		J_{t_{n}}(x, \bm{u}) = \mathbb{E}_{t_{n},x}\sBrackets{\Phi\rBrackets{\sum_{k=n}^{N}H\Big(t_{k}, X_{t_{k}}^{\cu}, \bm{u}_{t_{k}}\big(X_{t_{k}}^{\cu}\big), t_{n}\Big)}},
	\end{split}
\end{equation}
where the arguments of $J$ and $H$ here have been adjusted, vis-à-vis \eqref{eq:KMpreferences}, to accomodate for the dependence on the controlled state process $X^{\cu}$.\footnote{In presenting their results, \cite{BjoerkKhapkoMurgoci2021:TICT} start with a more parsimonious model where the function $H$ does not depend on the present time $t_{n}$; for example, one can think about the discount factor in \eqref{eq:CRRA-CES_preferences} being equal to $\delta^{t_{k}}$ instead of $\delta^{t_{k}-t_{n}}$. In the interest of brevity, we directly state the complete case with present-time dependence.}

As mentioned before, the nonlinearity of $\Phi$ induces the solution of $\sup_{\cu \in \bm{\mathcal{A}}} J_{t_{n}}(x, \bm{u})$ to be time-inconsistent.\footnote{Generally speaking, except for some specific cases, the dependence of $H$ on the current time $t_{n}$ and state $x$ is also a source of time-inconsistency, independently of the form of $\Phi$.} That is, an optimal strategy for $J_{t_{n}}(x, \bm{u})$ is not guaranteed to be optimal when applied at any future subinterval $\{t_{k},\dots, T\}$, with $t_{k} \geq t_{n}$, and so the Bellman's principle of optimality fails to hold globally. 

Building on their previous contributions to equilibrium control theory, where a similar path is taken, \cite{BjoerkKhapkoMurgoci2021:TICT} then look for a Nash subgame-perfect equilibrium control, as specified in the following definition (cf. Definition 5.2 in \cite{BjoerkKhapkoMurgoci2021:TICT}).
\begin{definition} \label{def:DT_EquilibriumControl}
	Fix an arbitrary $t_{n} \in \{0,t_1,\dots,T\}, \; x \in \mathcal{X}$, and a pair of controls $\cu, \ecu \in \bm{\mathcal{A}}$. Now define a new control $\cu^{t_{n}} = \{\cu_{t_{k}}^{t_{n}}\}_{k=n}^{N}$ by setting
	\begin{equation}
		\cu_{t_{k}}^{t_{n}}(x) = \left\{
		\begin{aligned}
			&\ecu_{t_{k}}(x),\quad \text{for } t_{k} \in \{t_{n+1},\dots,T\}, \\
			&\cu_{t_{n}}(x),\quad \text{for } t_{k} = t_{n}.
		\end{aligned}
		\right.
	\end{equation}
	If, for every fixed  $t_{n} \in \{0,t_1,\dots,T\}, \; x \in \mathcal{X}$, it holds that
	\begin{equation}
		\sup_{\cu_{t_{n}}(x) \in \mathcal{A}(t_{n},x)} J_{t_{n}}\left(x,\cu^{t_{n}}\right) = J_{t_{n}}\left(x, \ecu\right),
	\end{equation}
	then $\ecu$ is referred to as an equilibrium control law.
	
	In relation, for an admissible equilibrium control law $\ecu$ --  assuming that it exists -- we define the equilibrium value function $\widehat{V} = \{\widehat{V}_{t_{n}}\}_{n=0}^{N}$ as
	\begin{equation*}
		\begin{split}
			\widehat{V}_{t_{n}}(x) := J_{t_{n}}(x,\ecu).
		\end{split}
	\end{equation*}
\end{definition}

Similarly to the standard Bellman equation for time-consistent problems, the aim is to recursively characterize the value function $\widehat{V}$. To do so, we need first to introduce an auxiliary function $ f^{\cu} = \{ f_{t_{n}}^{\cu} \}_{n=0}^{N}$, defined as follows: 
\begin{equation}\label{eq:def_f_n_cu_sequence}
	f_{t_{n}}^{\cu}(x,z,t_{m}) = \mathbb{E}_{t_{n},x}\sBrackets{\Phi\rBrackets{\sum_{k=n}^{N}H\Big(t_{k}, X_{t_{k}}^{\cu}, \cu_{t_{k}}\big(X_{t_{k}}^{\cu}\big), t_{m}\Big) + z}},
\end{equation}
for any $(x,z,t_{m}) \in \mathcal{X} \times \mathbb{R} \times \{0,t_{1},\dots,t_{n}\}$ and  $\cu \in \bm{\mathcal{A}}$. In particular, note that we have
\begin{align*}
	f_{t_{n}}^{\cu}(x,0,t_{n}) & = J_{t_{n}}(x, \cu),\\
	f_{t_{n}}^{\ecu}(x,0,t_{n}) & = \widehat{V}_{t_{n}}(x).
\end{align*}

We are now ready to state the central result (cf.  Proposition 6.3 in \cite{BjoerkKhapkoMurgoci2021:TICT}) of the discrete-time equilibrium theory for KM preferences, which is the recursion for the equilibrium value function $\widehat{V}$ and the auxiliary function $f^{\ecu}$. Together, these form the so-called \textit{extended Bellman system}.
\begin{proposition}
	\label{prop:DT_KM_extended_Bellman_system}
	Let the reward functional $J$ be of the form \eqref{eq:KM_DT_reward_func}, and let $f^{\cu}$ be defined as in \eqref{eq:def_f_n_cu_sequence}. For any $t_{n} \in \{0,t_{1},\dots,t_{N - 1}\}$ and $(x,z,t_{m}) \in \mathcal{X} \times \mathbb{R} \times \{0,t_{1},\dots,t_{n}\}$, the extended Bellman system is given by
	\begin{equation}
		\begin{split}
			\label{KM_DT_extendedBellman}
			\widehat{V}_{t_{n}}(x) & =  \sup_{\cu_{t_{n}}(x) \in \mathcal{A}(t_{n},x)}\mathbb{E}_{t_{n},x}\Big[f_{t_{n+1}}^{\ecu}\rBrackets{X^{\cu^{t_{n}}}_{t_{n+1}},H\left(t_{n},x,\cu_{t_{n}}(x),t_{n}\right),t_{n}}\Big] ,\\ 
			\widehat{V}_{T}(x) & = \sup_{\cu_{T}(x) \in \mathcal{A}(T,x)} \, \Phi \big(H(T,x,\cu_{T}(x),T)\big), \\
			f_{t_{n}}^{\ecu}(x,z,t_{m}) &=  \, \mathbb{E}_{t_{n},x}\sBrackets{f_{t_{n+1}}^{\ecu}\left(X^{\ecu}_{t_{n+1}},H\left(t_{n},x,\ecu_{t_{n}}(x),t_{m}\right) + z,t_{m}\right)}, \\
			f_{T}^{\ecu}(x,z,t_{m})&= \Phi\Big(H\big(T,x, \ecu_T(x),t_{m}\big) + z\Big). \\
		\end{split}
	\end{equation}
\end{proposition}
\begin{remark}\label{rem:Bjoerk_results_on_KM_in_DT}
	Let us look at some special cases to build intuition about the system \eqref{KM_DT_extendedBellman}. For instance, consider a discount function $\Delta : \mathcal{T}^{2} \to [0,1]$, with $\Delta(\iota,\iota) = 1$, and assume that $\Phi(y) = y$ and $H$ (with abuse of notation) takes the separable form 
	\begin{equation*}
		H\Big(t_{k}, X_{t_{k}}^{\cu}, \cu_{t_{k}}\big(X_{t_{k}}^{\cu}\big), t_{n}\Big) = \Delta\left(t_{k}, t_{n}\right) H\Big(X_{t_{k}}^{\cu}, \cu_{t_{k}}\big(X_{t_{k}}^{\cu}\big)\Big).
	\end{equation*}
	This yields the following (time-additive) expected utility functional:
	\begin{equation*}
		J_{t_{n}}(x, \bm{u}) = \mathbb{E}_{t_{n},x}\sBrackets{\sum_{k=n}^{N}\Delta\left(t_{k}, t_{n}\right) H\Big(X_{t_{k}}^{\cu}, \cu_{t_{k}}\big(X_{t_{k}}^{\cu}\big)\Big)}.
	\end{equation*}
	In this setting, for $t_{n} \in \{0,t_{1},\dots,t_{N - 1}\}$ and $(x,t_{m}) \in \mathcal{X} \times \{0,t_{1},\dots,t_{n}\}$, the extended Bellman system is given by
	\begin{equation} \label{DT_extendedBellmanGeneralDiscount}
		\begin{split}
			\widehat{V}_{t_{n}}(x) & = \sup_{\cu_{t_{n}}(x) \in \mathcal{A}(t_{n},x)} 
			\left\{ H(x,\cu_{t_{n}}(x)) + \mathbb{E}_{t_{n},x}\Big[\bar{f}_{t_{n+1}}^{\ecu}\rBrackets{X^{\cu^{t_{n}}}_{t_{n+1}},t_{n}}\Big] \right\},\\
			\widehat{V}_{T}(x) & = \sup_{\cu_{T}(x) \in \mathcal{A}(T,x)} \, H(x,\cu_{T}(x)), \\
			\bar{f}^{\ecu}_{t_{n}}(x,t_{m}) & =  \, \Delta\left(t_{n}, t_{m}\right)H(x,\ecu_{t_{n}}(x)) + \mathbb{E}_{t_{n},x}\sBrackets{\bar{f}^{\ecu}_{t_{n+1}}\left(X^{\ecu}_{t_{n+1}},t_{m}\right)}, \\
			\bar{f}^{\ecu}_{T}(x,t_{m})&= \Delta\left(T,t_{m}\right)H\big(x, \ecu_T(x)\big), \\
		\end{split}
	\end{equation}
	where $\bar{f}^{\ecu}_{t_n}(x, t_m) := f^{\ecu}_{t_n}(x, 0, t_m)$ and it has the probabilistic representation
	\begin{equation*}
		\bar{f}^{\ecu}_{t_{n}}(x,t_{m}) = \mathbb{E}_{t_{n},x}\sBrackets{\sum_{k=n}^{N}\Delta\left(t_{k}, t_{m}\right) H\Big(X_{t_{k}}^{\ecu}, \ecu_{t_{k}}\big(X_{t_{k}}^{\ecu}\big)\Big)}.
	\end{equation*}
	
	The key distinction between \eqref{KM_DT_extendedBellman} and \eqref{DT_extendedBellmanGeneralDiscount} is that, in the case of KM preferences, the auxiliary function $f^{\cu}$ is defined in terms of an additional argument $z$. This variable is instrumental in handling the utility induced by present-time actions (which cannot be factored out of the conditional expectation due to non-time-additivity) at later times.\footnote{In the continuous-time model, whereby \eqref{KM_DT_extendedBellman} translates into a coupled system of two PDEs, we will show that for some choices of the preference components, one can suppress the variable $z$ by studying an associated countably infinite system of PDEs.}
	
	When $\Delta\left(t_{k}, t_{n}\right) = \delta^{t_{k}-t_{n}}$ (exponential discounting), we obtain the classic recursion:
	\begin{equation*}
		\begin{split}
			\widehat{V}_{t_{n}}(x) & = \sup_{\cu_{t_{n}}(x) \in \mathcal{A}(t_{n},x)} 
			\left\{ H\left(x,\cu_{t_{n}}(x)\right) + \delta^{t_{n+1}-t_{n}} \, \mathbb{E}_{t_{n},x}\Big[\widehat{V}_{t_{n+1}}\rBrackets{X^{\cu^{t_{n}}}_{t_{n+1}}}\Big] \right\},\\
			\widehat{V}_{T}(x) & = \sup_{\cu_{T}(x) \in \mathcal{A}(T,x)} \, H\left(x,\cu_{T}(x)\right). 
		\end{split}
	\end{equation*}
	
\end{remark}

We can now present our continuous-time equilibrium theory for KM preferences.

\section{From discrete to continuous time}\label{sec:general_theory}

Section \ref{subsec:setup} introduces the setup, including some useful definitions and preliminaries. In Section \ref{subsec:informal_arguments}, we follow the approach of \cite{BjoerkKhapkoMurgoci2021:TICT} and reach via heuristic arguments an extension of the standard Hamilton\textendash Jacobi\textendash Bellman partial differential equation that characterizes the value function and optimal control of time-consistent problems. This informal derivation should serve as a motivation for (and intuition about) the verification theorem, which we present formally in \mbox{Section \ref{subsec:verification_theorem}.} 

\subsection{Setup} \label{subsec:setup}
Consider a probability space $\rBrackets{\Omega, \mathcal{F}, \mathbb{P}}$, endowed with a (right-continuous, increasing, and augmented) filtration $\bm{\mathcal{F}}:= \rBrackets{\mathcal{F}_t}_{\tin}$ generated by a Wiener process $W := (W_t)_{\tin}$. The controlled state process $X^{\cu} := (X_t^{\cu})_{\tin}$, taking values in $\XSet \subseteq \R$, solves the stochastic differential equation (SDE)
\begin{equation}\label{eq:SDE_X}
	dX_t^{\cu} = \mu(t, X_t^{\cu}, \cu(t, X_t^{\cu}))dt + \sigma(t, X_t^{\cu}, \cu(t, X_t^{\cu})) dW_t, \quad X_{0}^{\cu} = x_{0} \in \XSet,
\end{equation}
where  $\mu, \sigma :\TSet \times \XSet \times \R^d \rightarrow \R$ are continuous mappings representing the (controlled) drift and volatility, respectively, and $\cu: \TSet \times \XSet \rightarrow \R^d$ is the control law chosen by the agent, for some dimensionality $d \in \N$. We consider feedback strategies in the form of $
\cu(t, X^{\cu}_t)$, thus $\cu$ in general depends on the current time $t$ and on the value of the state process $X^{\cu}_{t}$. Whenever there is no confusion, we will write for brevity $\cu(t):=\cu(t, X^{\cu}_t)$.

The agent's reward functional is given by
\begin{equation}\label{eq:generalized_KM_CT_reward_func_with_t_dependence}
	J(t, x, \cu) = \EVtx \sBrackets{\Phi\rBrackets{\int \limits_{t}^{T}H(s, X_s^{\cu}, \bm{u}(s, X_s^{\cu}),t)\,ds + G \rBrackets{X_T^{\cu}, t} } }.
\end{equation}
Here $G: \XSet \times \TSet \rightarrow \R$ represents the discounted utility of terminal state, while $\Phi:\R \rightarrow \R$ and  $H:\TSet \times \XSet \times \R^d \times \TSet \rightarrow  \R$ carry their meaning from previous sections. The function $H$ is assumed to be continuous.

We proceed by outlining the conditions for the admissibility of a control \mbox{strategy}. 
\begin{definition}\label{def:CT_KM_without_tau_admissible_controls}
	A control $\cu$ is said to be admissible if, for all $(t, x) \in \TSet \times \XSet$, the following conditions hold:
	\begin{itemize}
		\item[(i)] $\cu(t, x) \in \mathcal{A}(t, x)$,
		where $\AMap:\TSet \times \XSet \rightarrow 2^{\R^{k}}$ is a continuous set-valued function representing the admissible values attained by $\cu(t,x)$.\footnote{More precisely, for each $(t, x) \in [0, T) \times \XSet$ and each $u \in \mathcal{A}(t, x)$, we assume that there exists a continuous control $\cu$ with $\cu(t, x) = u$.} 
		\item[(ii)] The SDE \eqref{eq:SDE_X} has a unique strong solution $X^{\cu}$.
		\item[(iii)] $J(t, x, \cu)$ is well defined and finite.
	\end{itemize}
	The set of admissible controls is denoted by $\ASet$.
\end{definition}

The agent searches for an equilibrium control using the following definition.
\begin{definition}\label{def:equilibrium_u}
	Consider a point $\tx$, a pair of controls $\ecu, \cu \in \ASet,$ and a real number $h \in (0, T - t]$. Define a new control $\cuh$ by setting 
	\begin{equation}\label{eq:perturbed_equilibrium_control}
		\cu_h(s,y) = \left\{
		\begin{aligned}
			&\cu(s,y),\quad \text{for } t \leq s < t + h,\quad y \in \XSet, \\
			&\ecu(s,y),\quad \text{for } t + h \leq s < T,\quad y \in \XSet.
		\end{aligned}
		\right.
	\end{equation}
	If the inequality
	\begin{equation}\label{eq:ecu_def_property}
		\liminf_{h \downarrow 0} \frac{J(t,x, \ecu) - J(t,x, \cuh)}{h} \geq 0
	\end{equation}
	holds for any $\cu \in \ASet$ and $(t,x) \in [0, T) \times \XSet$, then $\ecu$ is an (intrapersonal) equilibrium control. When $\ecu$ exists, the corresponding equilibrium value function $\widehat{V}$ is defined as
	\begin{equation}\label{eq:KM_CT_with_tau_equilibrium_VF}
		\begin{split}
			\widehat{V}(t, x) := J(t, x, \ecu).
		\end{split}
	\end{equation} 
\end{definition}

\begin{remark}
	This notion of equilibrium, which is based on the first-order expansion of the reward functional around the candidate $\ecu$ (thus often referred to as a \textit{weak equilibrium}), was initially formulated in \cite{EkelandPirvu2008:MFE} and \cite{EkelandLazrak2010:MFE}. 
	
	Two other definitions have emerged in the literature. The first is a natural extension of Definition \ref{def:DT_EquilibriumControl} in continuous time, and it has been studied in different settings by \cite{HuangZhou2021:MOR} and \cite{HeJiang2021:SICON}:
	\begin{itemize} 
		\item \textit{Strong equilibrium.} A control $\ecu \in \ASet$ is a strong equilibrium if for any $(t,x) \in [0,T) \times \mathcal{X}$ and $\cu\in \ASet,$ there exists $\epsilon_{0} \in (0,T-t)$ such that
		\begin{equation} \label{eq:StrongEquilibrium}
			J(t,x, \cuh) - J(t,x, \ecu)  \leq 0, \quad \mbox{for any } \epsilon \in (0,\epsilon_{0}],
		\end{equation}
		where $\cuh$ is defined in \eqref{eq:perturbed_equilibrium_control}.
	\end{itemize}
	In the same paper cited above, He and Jiang argued that a strong equilibrium may not exist under several sources of time-inconsistency. For this reason, they suggested the following less restrictive definition, which allows the agent to contemplate only alternative strategies that differ from $\ecu$ at a given time and state:
	\begin{itemize}
		\item \textit{Regular equilibrium.} A control $\ecu \in \ASet$ is a regular equilibrium if for any $(t,x) \in [0,T) \times \mathcal{X}$ and $\cu\in \ASet,$ with $\ecu(t,x) \neq \cu(t,x)$, there exists $\epsilon_{0} \in (0,T-t)$ such that \eqref{eq:StrongEquilibrium} holds. 
	\end{itemize}
	We leave the question of whether strategies that achieve such equilibria exist (or even coincide) under KM preferences for future study.
\end{remark}

Like in the discrete-time case, to characterize the equilibrium value function $\widehat{V}$ (ultimately as the solution of a system of PDEs), we bring in an auxiliary function $f^{\cu} : \TSet \times \mathcal{X} \times \mathbb{R} \times \mathcal{T} \to \mathbb{R}$.
\begin{definition}\label{def:KM_CT_with_tau_equilibrium_V_f}
	For any $\cu \in \ASet$,  $(t,x,z,\tau) \in \TSet \times \XSet \times \mathbb{R} \times [0, t]$, the function $f^{\cu}$ is defined by
	\begin{equation}\label{eq:KM_CT_with_tau_f_probabilistic_representation}
		f^{\cu}(t, x,z, \tau) = \EVtx \sBrackets{\Phi\rBrackets{\int \limits_{t}^{T}H(s, X_s^{\cu}, \cu(s),\tau)\,ds + G \rBrackets{ X^{\cu}_T,\tau} + z } }.
	\end{equation}
	In particular, note that
	\begin{align}
		f^{\cu}(t, x,0, t) &= J(t, x, \cu), \label{eq:KM_CT_with_tau_f_probabilistic_representation_z_0_tau_t}\\
		f^{\ecu}(t, x,0, t) &= J(t, x, \ecu) = \widehat{V}(t, x) \label{eq:KM_CT_with_tau_f_ecu_at_0}.
	\end{align}
	In what follows, if one of the coordinates of $f^{\cu}(t,x,z, \tau)$ is fixed, we place it in the superscript after a vertical bar. For instance, $f^{\cu|\tau}(t,x,z)$ indicates that $\tau$ takes some constant value, and the resulting function is seen as depending on $(t,x,z)$ only. The same principle applies to other functions.
\end{definition}

In Section \ref{subsec:EqTheory_kM_DT}, we discussed the role of the variable $z$ to ``keep track'' of the current utility $H$ at later points in time. With this intuition in mind, for an arbitrary but fixed $\tau \in [0,t]$, we introduce a stochastic process $Z^{\cu} =(Z^{\cu}_s)_{s \in [t, T]}$ characterized by the dynamics
\begin{equation}\label{eq:SDE_Z}
	dZ_{s}^{\cu} = H(s, X^{\cu}_s, \cu(s), \tau) ds,\qquad  Z^{\cu}_t = z.
\end{equation}

This leads to our definition of a differential operator (or infinitesimal generator) for the controlled state processes $(X^{\cu}, Z^{\cu})$.
\begin{definition}\label{def:differential_operator}
	Let $X^{\cu}$ be given by \eqref{eq:SDE_X}, $Z^{\cu}$ be given by \eqref{eq:SDE_Z}, and $\xi$ be a map from  $(t, x, z, \tau) \in \TSet \times \XSet \times \R \times \TSet$ to $\mathbb{R}$. Suppose $\xi \in \textgoth{C}^{1,2, 1, 1}\left(\TSet \times \XSet \times \R \times \TSet\right)$,\footnote{Given a positive integer $r$, $\textgoth{C}^{r}(\mathbb{D})$ indicates the space of functions that are continuously differentiable up to order $r$ on the domain $\mathbb{D}$. For functions of multiple variables, the order of continuous differentiability in each variable is listed in the superscript.} and denote by $\partial_y \, \xi(y,\cdot) $ and $\partial_{yy} \, \xi(y,\cdot)$ its first-order and second-order partial derivative in $y$, respectively. For any $\cu \in \ASet$, the controlled differential operator $\mathcal{D}^{\cu}$ applied to $\xi$ is defined as follows:
	\begin{equation}\label{eq:def_differential_operator}
		\begin{aligned}
			\mathcal{D}^{\cu}\xi(t,x,z,\tau) &=  \partial_t \xi(t,x,z,\tau) + \mu(t,x,\cu(t,x)) \partial_x \xi(t,x,z,\tau) \\
			&  \quad + \dfrac{1}{2}\rBrackets{\sigma(t,x,\cu(t,x))}^{2}\partial_{xx}\xi(t,x,z,\tau)\\
			& \quad + H(t, x, \cu(t,x), \tau) \partial_z \xi(t,x,z,\tau) + \partial_\tau \xi(t,x,z,\tau).
		\end{aligned}
	\end{equation}
	For a constant control $u$, the differential operator is denoted by $\mathcal{D}^{u}$ and defined analogously. 
\end{definition}

The next definition of an $\mathcal{L}^2$ function space deals with integrability conditions (cf. Definition 3.3 in \cite{Lindensjoe2019:ORL}).
\begin{definition}
	\label{def:admissible_space_of_functions}
	Consider an arbitrary control $\cu \in \ASet$. A function $ \xi:\TSet \times \XSet \times \mathbb{R} \times \TSet \rightarrow \mathbb{R}$ is said to belong to the space $\mathcal{L}^2(X^{\cu})$ if, for any $(t, x, z, \tau) \in \TSet \times \XSet \times \mathbb{R} \times [0, t]$, there exists a constant $\bar{h} \in (0, T - t)$ such that
	\begin{equation*}
		\begin{split}
			& \mathbb{E}_{t,x, z} \Bigg[ \sup_{0 \leq h \leq \Bar{h}} \Bigg\vert \int_{t}^{t+h}\dfrac{1}{h}\mathcal{D}^{\cu}\xi(s,X_{s}^{\cu}, Z_{s}^{\cu}, \tau) ds \; \Bigg\vert \Bigg.\\
			& \hspace{2cm} \Bigg. + \int_{t}^{t+\bar{h}}\Bigg(  \partial_{x} \xi(s,X_{s}^{\cu}, Z_{s}^{\cu}, \tau)\sigma \bigl(s,X_{s}^{\cu},\cu(s,X_{s}^{\cu})\bigr)\Bigg)^{2} ds \Bigg] < \infty,
		\end{split}
	\end{equation*}
	where $\mathbb{E}_{t,x,z}\left[\cdot\right]$ denotes the conditional expectation given $X_{t}^{\cu} = x$ and $Z^{\cu}_t = z$.
\end{definition}

In the end, we have the ensuing lemma. This serves a twofold aim: First, it gives a recursive representation of $f^{\cu}$. Second, it yields a Feynman\textendash Kac-type formula showing that $f^{\cu}$ solves an associated PDE. Both results will be useful later on.
\begin{lemma}\label{lem:CT_KM_with_tau_recursion_for_f}
	Assume that $f^{\cu} \in \mathcal{L}^2(X^{\cu})$. For any admissible control $\cu \in \ASet$, $(t,x,z,\tau) \in \TSetNoT \times \XSet \times \mathbb{R} \times [0, t]$, and $h \in [0, T - t]$, $f^{\cu}$ satisfies the recursion
	\begin{align}
		f^{\cu}(t, x, z, \tau)&= \EVtx\sBrackets{f^{\cu}\rBrackets{t + h, X^{\cu}_{t+h},  \inttth H(s, X_s^{\bm{u}}, \bm{u}(s), \tau) \, ds + z, \tau}}, \label{eq:CT_KM_with_tau_f_cu_recursion_rule} \\
		f^{\cu}(T, x, z, \tau)&= \Phi\rBrackets{G(x, \tau) + z}.\label{eq:CT_KM_with_tau_f_cu_recursion_terminal_condition}
	\end{align}
	In addition, if $f^{\cu|\tau} \in \textgoth{C}^{1,2, 1}\left(\TSet \times \XSet \times \R \right)$, then $f^{\cu|\tau}$ solves
	\begin{equation}\label{eq:CT_KM_with_tau_PDE_f_cu}
		\mathcal{D}^{\cu}f^{\cu|\tau}(t,x,z) = 0.
	\end{equation}
\end{lemma}

\subsection{Informal arguments}\label{subsec:informal_arguments}
In discrete time, Proposition \ref{prop:DT_KM_extended_Bellman_system} provided us with the recursive program for the equilibrium value function $\widehat{V}$ and auxiliary function $f^{\ecu}$. In what follows, allowing for some degree of informality,  we evince that the limit of this recursion -- as the time-discretization step goes to zero -- is given by a system of PDEs of HJB type. In Section \ref{subsec:verification_theorem}, we confirm that this is the correct system to study.\footnote{For a different approach to establish an extended dynamic programming principle for time-inconsistent problems via backward SDEs, see \cite{HernandezPossamai2023:AAP}.}

First, let us recall the discrete-time recursion from \eqref{KM_DT_extendedBellman}:
\begin{equation*}
	\widehat{V}_{t_{n}}(x) = \sup_{u \in \mathcal{A}(t_{n},x)}\mathbb{E}_{t_{n},x}\sBrackets{f_{t_{n+1}}^{\ecu}\rBrackets{X^{\cu^{t_{n}}}_{t_{n+1}},H(t_{n},x,u,t_{n}),t_{n}}},
\end{equation*}
where we use the notation $u := \cu_{t_{n}}(x)$. We can equivalently write the above as
\begin{equation*}
	\begin{split}
		\widehat{V}_{t_{n}}(x) =& \sup_{u \in \mathcal{A}(t_{n},x)}\bigg\{ \mathbb{E}_{t_{n},x}\sBrackets{\widehat{V}_{t_{n+1}}(X_{t_{n+1}}^{\cu^{t_{n}}})} + \mathbb{E}_{t_{n},x}\sBrackets{f_{t_{n+1}}^{\ecu}\rBrackets{X^{\cu^{t_{n}}}_{t_{n+1}},H(t_{n},x,u,t_{n}),t_{n}}} \bigg. \nonumber \\
		& \hspace{2.cm} \bigg. - \mathbb{E}_{t_{n},x}\sBrackets{f_{t_{n+1}}^{\ecu}\rBrackets{X^{\cu^{t_{n}}}_{t_{n+1}},0,t_{n+1}}} \bigg\},
	\end{split}
\end{equation*}
by using the fact that $\mathbb{E}_{t_{n},x}\sBrackets{\widehat{V}_{t_{n+1}}(X_{t_{n+1}}^{\cu^{t_{n}}})} = \mathbb{E}_{t_{n},x}\sBrackets{f_{t_{n+1}}^{\ecu}\rBrackets{X^{\cu^{t_{n}}}_{t_{n+1}},0,t_{n+1}}}$. This implies that, for any $\cu^{t_n}$, the following inequality holds:
\begin{equation*}
	\begin{split}
		& 0 \geq  \mathbb{E}_{t_{n},x}\sBrackets{\widehat{V}_{t_{n+1}}(X_{t_{n+1}}^{\cu^{t_n}})} - \widehat{V}_{t_{n}}(x) \\
		& \qquad + \mathbb{E}_{t_{n},x}\sBrackets{f_{t_{n+1}}^{\ecu}\rBrackets{X^{\cu^{t_n}}_{t_{n+1}},H(t_{n},x,\cu^{t_n}_{t_n}(x),t_{n}),t_{n}}}  \nonumber - \mathbb{E}_{t_{n},x}\sBrackets{f_{t_{n+1}}^{\ecu}\rBrackets{X^{\cu^{t_n}}_{t_{n+1}},0,t_{n+1}}}.
	\end{split}
\end{equation*}
Now take an arbitrary but fixed point $(t, x) \in [0, T) \times \XSet$ and consider a control $\cu_h \in \ASet$ as in \eqref{eq:perturbed_equilibrium_control}. Writing $V$ instead of $\widehat{V}$, $\cu_{h}$ instead of $\cu^{t_n}$, and replacing $t_n$ by $t$ and $t_{n+1}$ by $t + h$, with $h = t_{n+1} -  t_{n}$, we reformulate the above inequality as follows:
\begin{flalign}
	& 0 \geq \EV_{t,x} \sBrackets{V(t+h,X_{t+h}^{\cuh})} - V(t,x) \notag \\
	& \qquad + \EV_{t,x} \sBrackets{f \rBrackets{t+h, X^{\cuh}_{t+h}, \int \limits_{t}^{t+h} H(s,X^{\cuh}_{s}, \cuh(s), t)\,ds, t}}  \notag \\
	& \qquad  - \EV_{t,x} \left[f \rBrackets{t+h, X^{\cuh}_{t+h}, 0, t + h}\right], \notag
\end{flalign}
where the time index has been moved inside the parentheses as an argument. Adding and subtracting the term $\EV_{t,x} \sBrackets{f \rBrackets{t+h, X^{\cuh}_{t+h}, 0, t}}$, dividing the above inequality by $h>0$, and taking the limit as $h \downarrow 0$, we anticipate to obtain
\begin{flalign*}
	0 & \geq \mathcal{D}^{u}V(t,x) + \partial_{z} f\rBrackets{t, x, 0, t} H(t, x, u, t) - \partial_{\tau} f(t, x, 0, t).
\end{flalign*}
Taking the supremum with respect to $u \in \mathcal{A}(t,x)$ over the infinitesimal time period $[t, t+ h)$, we write
\begin{align*}
	0 & \geq \sup_{u \in \mathcal{A}(t,x)}\left\{\mathcal{D}^{u}V(t,x) + \partial_{z} f\rBrackets{t, x, 0, t} H(t, x, u, t) - \partial_{\tau} f(t, x, 0, t) \right\},
\end{align*}
which becomes an equality at the equilibrium control $\ecu(t,x)$:
\begin{align*}
	0 & = \mathcal{D}^{\ecu}V(t,x) + \partial_{z} f\rBrackets{t, x, 0, t} H(t, x, \ecu(t, x), t) - \partial_{\tau} f(t, x, 0, t).
\end{align*}
The terminal condition for $V$ is given by $
V(T, x) = \Phi\rBrackets{G(x,T)}$. In addition, an application of Lemma \ref{lem:CT_KM_with_tau_recursion_for_f} for $\cu = \ecu$ yields the PDE and terminal condition for $f$.

We summarize these derivations below.
\paragraph*{The extended HJB system.} Let the reward functional $J$ be of the form \eqref{eq:generalized_KM_CT_reward_func_with_t_dependence}. For any $(t,x,z,\tau) \in [0, T) \times \XSet \times \mathbb{R} \times [0, t]$, the extended HJB system is given by
\begin{align}
	\hspace{-1cm} 0 & = \sup_{u \in \mathcal{A}(t,x)}\left\{\mathcal{D}^{u}V(t,x) + \partial_{z} f(t, x, 0, t) H(t, x, u, t) - \partial_{\tau} f(t, x, 0, t)  \right\},\label{eq:CT_KM_with_tau_EHJB_PDE_V_t_x_z} \tag{$S1$}\\
	\hspace{-1cm}  0 & =  \mathcal{D}^{\ecu}f^{|\tau}(t, x, z),
	\label{eq:CT_KM_with_tau_EHJB_PDE_f_t_x_z} \tag{$S2$}\\
	& V(T,x)  = \Phi\rBrackets{G\rBrackets{x, T}}, \label{eq:CT_KM_with_tau_EHJB_PDE_V_T_x} \tag{$S3$}\\ 
	\hspace{-0.3cm}  & f(T, x, z, \tau) = \Phi\rBrackets{G\rBrackets{x, \tau} + z}. \tag{$S4$}\label{eq:CT_KM_with_tau_EHJB_PDE_f_T_x_z_tau} 
\end{align}

\subsection{Verification theorem} \label{subsec:verification_theorem}

The following theorem verifies that, under suitable regularity assumptions, the candidate functions $f$ and $V$ solving the extended HJB system \eqref{eq:CT_KM_with_tau_EHJB_PDE_V_t_x_z}-\eqref{eq:CT_KM_with_tau_EHJB_PDE_f_T_x_z_tau} characterize the equilibrium value function, and that the argument of the supremum in \eqref{eq:CT_KM_with_tau_EHJB_PDE_V_t_x_z} is an equilibrium control.

\begin{theorem}\label{th:verification_theorem_with_tau}
	Assume that the following conditions are satisfied:
	\begin{enumerate}[itemindent=0.4cm, itemsep=-0.5ex]
		\item[(C1)] An admissible equilibrium control $\ecu$ exists and realizes the $\sup$ \mbox {in \eqref{eq:CT_KM_with_tau_EHJB_PDE_V_t_x_z}.}
		\item[(C2)] $V(t,x)$ and $f(t,x,z, \tau)$ solve the extended HJB system  \eqref{eq:CT_KM_with_tau_EHJB_PDE_V_t_x_z}-\eqref{eq:CT_KM_with_tau_EHJB_PDE_f_T_x_z_tau}.
		\item[(C3)] $V\in \textgoth{C}^{1,2}\left(\TSet \times \XSet\right)$ and $f\in \textgoth{C}^{1,2, 1, 2}\left(\TSet \times \XSet \times \R \times \TSet\right)$.
		\item[(C4)] $V, f \in \mathcal{L}^2(X^{\cu})$ for any $\cu \in \ASet$.
		\item[(C5)] For any $\cu \in \ASet$, there exists $\overline{h} > 0$ such that  
		\begin{equation*}
			\sup_{h \in (0, \overline{h}), \, \eta : \Omega \to [t,t+h]}\mathbb{E}_{t,x,0}\bigg[ \big\vert \partial_{z} f(t + h, X^{\cu}_{t+h}, Z^{\cu|t}_{\eta}, t)H(\eta, X^{\cu}_{\eta}, \cu(\eta, X^{\cu}_{\eta}), t) \big\vert \bigg] < \infty.
		\end{equation*} 
		\item[(C6)] For any $\cu \in \ASet$, there exists $ \overline{h} > 0$ such that 
		\begin{equation*}
			\sup_{h \in (0, \overline{h}), \, \iota: \Omega \rightarrow [t, t+h]} \EVtx\bigg[\big\vert \partial_{\tau} f(t + h, X^{\cu}_{t+h}, 0, t) \big\vert + \big\vert f_{\tau \tau }(t + h, X^{\cu}_{t+h}, 0, \iota) h \big\vert\bigg]< \infty.
		\end{equation*}
	\end{enumerate}
	Then:
	\begin{enumerate}[itemindent=0.4cm, itemsep=-0.5ex]
		\item[(R1)] $f(t, x, z, \tau) = f^{\ecu}(t, x, z, \tau)$ and has the probabilistic representation \eqref{eq:KM_CT_with_tau_f_probabilistic_representation}.
		\item[(R2)] $V(t,x) = J(t,x,\ecu)$ for $\ecu$ realizing the sup in \eqref{eq:CT_KM_with_tau_EHJB_PDE_V_t_x_z}.
		\item[(R3)] $\ecu$ is an equilibrium control in the sense of Definition \ref{def:equilibrium_u}.
		\item[(R4)] $\widehat{V}(t,x) = V(t, x)$ is the equilibrium value function and has the probabilistic representation \eqref{eq:KM_CT_with_tau_equilibrium_VF}.
	\end{enumerate}
\end{theorem}

\begin{remark}
	Let us comment on the assumptions of the theorem.  (C1) and (C2) are equivalent to standard first-order conditions for optimality. (C3) is a differentiability requirement on $V$ and $f$ to apply It\^{o}'s lemma. (C4)-(C6) are sufficient integrability conditions on $V, f$ and their derivatives under which the dominated convergence theorem can be applied. In addition, (C4) ensures that the relevant stochastic integrals are martingales with expectation zero within our setting.
\end{remark}

\begin{remark}
	\label{rem:KM_CT_comparison_to_cases_from_literature} In parallel with Remark \ref{rem:Bjoerk_results_on_KM_in_DT}, we relate the system \eqref{eq:CT_KM_with_tau_EHJB_PDE_V_t_x_z}-\eqref{eq:CT_KM_with_tau_EHJB_PDE_f_T_x_z_tau} to some special cases. Consider a discount function $\Delta : \mathcal{T}^{2} \to [0,1]$, with $\Delta(\iota,\iota) = 1$, and assume that $\Phi(y) = y$ and that $H$ and $G$ (again with abuse of notation) take the separable form 
	\begin{equation} \label{eq:CT_RewardFunctional_NonExponentialDiscounting}
		\begin{split}
			H(t, x, u, \tau) & = \Delta(t,\tau)H(x,u), \\
			G(x, \tau) & = \Delta(T,\tau) G(x).
		\end{split}
	\end{equation}
	In this case, the objective function reads as
	\begin{equation*}
		J(t, x, \cu) = \EVtx \sBrackets{\int \limits_{t}^{T}\Delta(s,t)H(X_s^{\cu}, \bm{u}(s,X_s^{\cu}))\,ds + \Delta(T,t)G \rBrackets{X_T^{\cu}} } ,
	\end{equation*}
	and the extended HJB system  is given by
	\begin{equation}
 \label{eq:ExtendedHJBNonExponentialDiscounting}
 \begin{split}
			\hspace{-1cm} 0 & = \sup_{u \in \mathcal{A}(t,x)}\left\{\mathcal{D}^{u}V(t,x) + H(x,  u) - \partial_{\tau} \bar{f}(t, x, t)  \right\}, \\
			\hspace{-1cm}  0 & =  \mathcal{D}^{\ecu}\bar{f}^{\, | \, \tau}(t, x),\\
			& V(T,x)  = G(x), \\ 
			\hspace{-0.3cm}  & \bar{f}(T, x, \tau) = \Delta(T,\tau)G(x), 
   \end{split}
	\end{equation}
	where $\bar{f}(t,x,\tau) := f(t,x,0,\tau)$.
 
	Clearly, when $\Delta\left(t, \tau\right) = e^{-\delta(t-\tau)}$, we have $\partial_{\tau} \bar{f}(t, x, t) = \delta  \bar{f}(t, x, t) = \delta V(t,x)$ (which follows from the definition of $\bar{f}$ and \eqref{eq:KM_CT_with_tau_f_probabilistic_representation}), and we retrieve the classic HJB equation
	\begin{equation*}
		\begin{split}
			\hspace{-1cm} 0 & = \sup_{u \in \mathcal{A}(t,x)}\left\{\mathcal{D}^{u}V(t,x) + H(x, u) -\delta V(t,x) \right\},\\
			& V(T,x)  = G(x).
		\end{split}
	\end{equation*}
\end{remark}

We have already commented that, under KM utility, the auxiliary function $f$ takes an additional argument $z$. Comparing \eqref{eq:CT_KM_with_tau_EHJB_PDE_V_t_x_z} to the first equation in \eqref{eq:ExtendedHJBNonExponentialDiscounting}, the novelty lies in the term $\partial_{z}f(t,x,0,t)H(t,x,u,t)$. In essence, the additional source of time-inconsistency induced by the nonlinearity of $\Phi$ is encoded in the $z$-derivative of $f$ at $z=0$, which becomes there an adjustment factor of the instantaneous utility $H$.\footnote{When $\Phi$ is the identity function, i.e. $\Phi(y) = y$ as in \eqref{eq:CT_RewardFunctional_NonExponentialDiscounting}, we have $\partial_{z}f(t,x,0,t) = 1$.}

In some instances, it is possible to derive a characterization of the extended HJB system in which $f$ assumes a simpler form. This is the subject of the next section (specifically, see Corollary \ref{cor:CT_KM_with_cur_t_dependence_EHJB_system_infinite}), where we focus on CRRA-CES preferences.\footnote{We illustrate a similar procedure in Appendix \ref{app:NoDependenceOnCurrentTime}, where we suppose that the agent's utility functions $H$ and $G$ do not depend on the current time $\tau = t$. }

\section[Application: Consumption-investment with CRRA-CES preferences and exponential discounting]{Application: Consumption-investment with \\ CRRA-CES preferences and exponential discounting} \label{sec:CRRA-CES_utility_exponential_discounting}
We consider a consumption-investment problem for an agent with exponentially discounted KM preferences with constant relative risk aversion and constant elasticity of intertemporal substitution. Notably, we show that in this case the equilibrium value function is separable in wealth and time and can be characterized as the solution of a \mbox{system of ODEs.} 

The financial market consists of a risk-free money market account $B$ and a (non-dividend paying) stock $S$. At time $t \in \TSet$, price dynamics are given by
\begin{equation*}
	\begin{split}
		& d B_t = B_t r dt, \quad B_0 = 1, \\
		& dS_t =  S_t\left( r+\lambda  \right) dt + S_t \sigma d W_t,  \quad S_0 = s_0 \in \mathbb{R}^{+},
	\end{split}
\end{equation*}
where $r, \lambda, \sigma$ are positive constants -- canonically interpreted as the risk-free rate, the risk premium, and the standard deviation of the stock return, respectively -- and $W$ is a one-dimensional Wiener process. The agent's decisions at time $t$ regarding consumption and investment are described in feedback form by the vector $\cu(t,x) := \left(\pi(t,x), c(t,x)\right)$, with $\pi(t,x)$ and $c(t,x)$ denoting the fraction of wealth invested in the risky asset and the consumption rate, respectively. For simplicity, we do not impose constraints on the optimal policies. The controlled wealth process $\left(X_{t}^{\pi,c}\right)_{t\in \TSet}$ then solves
\begin{equation*}
	\begin{split}
		dX^{\pi,c}_{t} & =  \big( X^{\pi,c}_{t}\left(r + \lambda\pi(t,X^{\pi,c}_{t})\right) -c(t,X^{\pi,c}_{t}) \big)dt + X^{\pi,c}_{t}\pi(t,X^{\pi,c}_{t})\sigma dW_{t},\\
		X^{\pi,c}_{0} & = x_{0} \in \mathbb{R}^{+}.
	\end{split}
\end{equation*}

In the notation of previous sections, the agent's preferences are described by
\begin{equation*}
	\begin{split}
		\Phi(x) & = \dfrac{1}{1-\alpha}x^{1-\alpha}, \\
		H(s,x,(\pi,c),t) & = e^{-\delta(s-t)}c^{\rho}, \\
		G(x,t) & = e^{-\delta(T-t)}x^{\rho}, \\
	\end{split}
\end{equation*}
for constant parameters $\alpha \geq 0, \, \delta \in [0,1], \, \rho < 1$. As mentioned, $\alpha$ identifies the relative risk aversion, $\delta$ is a discount rate, and $\rho$ specifies (but is not equal to) the elasticity of substitution. To avoid unnecessary complications, we suppose $\alpha \neq 1$ and postpone the case of $\alpha = 1$ to Appendix \ref{subec:Example_CRRA_unitaryRRA}.

The reward functional is thus
\begin{equation}\label{eq:CRRA-CES_RewardFunctional_KM_CT}
	J(t, x, (\pi, c)) = \mathbb{E}_{t,x}\sBrackets{\frac{1}{1 - \alpha}\rBrackets{\int \limits_{t}^{T} e^{-\delta \rBrackets{s - t}} \left(c(s)\right)^{\rho}\,ds + e^{-\delta \rBrackets{T - t}}\rBrackets{X^{\pi, c}_T}^{\rho}}^{\frac{1 - \alpha}{\rho}}},
\end{equation}
with the equilibrium value function defined as $\widehat{V}(t,x) := J(t, x, (\hat{\pi}, \hat{c}))$.

The extended HJB system for the decision maker with reward functional \eqref{eq:CRRA-CES_RewardFunctional_KM_CT} is given by
\begin{equation} \label{eq:EHJBsystem_CRRA}
	\begin{split}
		\displaystyle
		0 & = \sup_{(\pi, c) \in \mathcal{A}(t,x)}\Big\{\partial_{t} V(t,x) + \partial_x V(t, x) (x(r + \pi \lambda) - c) + \dfrac{1}{2}\partial_{xx} V(t, x) \sigma^2 \pi^2 x^2  \Big. \\
		& \hspace{2.5cm} \Big. + \partial_{z} f(t, x, 0, t) c^\rho- \partial_{\tau} f(t, x, 0, t) \Big\} , \\
		0 & = \mathcal{D}^{\hat{\pi},\hat{c}  }f^{\vert \tau}(t,x,z), \\
		& \quad V(T,x)  = \dfrac{1}{1 - \alpha} x^{1 - \alpha},\quad \\
		& \quad f(T,x,z,\tau)  = \dfrac{1}{1 - \alpha}\rBrackets{e^{-\delta \rBrackets{T - \tau}}  x^{\rho} + z}^{\frac{1 - \alpha}{\rho}},
	\end{split}
\end{equation}
and the probabilistic representation of $V$ and $f$ is given by
\begin{align*}
	V(t,x) &= \EVtx\sBrackets{\frac{1}{1 - \alpha}\rBrackets{\int \limits_{t}^{T}e^{-\delta \rBrackets{s - t}}\left(\hat{c}(s)\right)^{\rho}\,ds + e^{-\delta \rBrackets{T - t}}\rBrackets{X^{\hat{\pi}, \hat{c}}_T}^{\rho}}^{\frac{1 - \alpha}{\rho}}},\\
	f(t,x,z,\tau) &= \EVtx\sBrackets{\frac{1}{1 - \alpha}\rBrackets{\int \limits_{t}^{T}e^{-\delta \rBrackets{s - \tau}}\left(\hat{c}(s)\right)^{\rho}\,ds + e^{-\delta \rBrackets{T - \tau}}\rBrackets{X^{\hat{\pi}, \hat{c}}_T}^{\rho} + z}^{\frac{1 - \alpha}{\rho}} }.
\end{align*}

The above system appears challenging to work with, mainly due to the nature of the function $f$. Fortunately, we can obtain an alternative, more tractable form.

\begin{corollary}\label{cor:CT_KM_with_cur_t_dependence_EHJB_system_infinite}
	The extended HJB system \eqref{eq:EHJBsystem_CRRA} can be written in the form
	\begin{equation} \label{eq:KM_CT_PDE_for_V}
		\begin{split}
			\displaystyle
			0 & = \sup_{(\pi,c) \in \mathcal{A}(t,x)}\Big\{\partial_{t} V(t,x) + \partial_x V(t, x) (x(r + \pi \lambda) - c) + \dfrac{1}{2}\partial_{xx} V(t, x) \sigma^2 \pi^2 x^2  \Big. \\ 
			& \hspace{2.5cm}\Big. + \dfrac{1}{\rho}\widetilde{V}^{(1)}(t,x) c^\rho- \delta \dfrac{1 - \alpha}{\rho} V(t,x) \Big\}, \\
			0 & = \widetilde{V}_{t}^{(k)}(t,x) + \partial_x\widetilde{V}^{(k)}(t, x) (x(r + \hat{\pi} \lambda) - \hat{c}) + \dfrac{1}{2}\partial_{xx}\widetilde{V}^{(k)}(t, x) \sigma^2 \hat{\pi}^2 x^2 \\
			& \qquad + \rBrackets{\dfrac{1 - \alpha}{\rho} - k}\widetilde{V}^{(k+1)}(t,x) \hat{c}^{\, \rho}  - \delta\rBrackets{\dfrac{1 - \alpha}{\rho} - k}\widetilde{V}^{(k)}(t,x), \\
			& \hspace{0.5cm} V(T,x)  = \dfrac{1}{1 - \alpha} x^{1 - \alpha} ,\\
			& \hspace{0.5cm} \widetilde{V}^{(k)}(T,x)  = x^{1 - \alpha - k\rho}.
		\end{split}
	\end{equation}
	In addition, the probabilistic representation of $\widetilde{V}^{(k)}$ is as follows:
	\begin{align*}
		\widetilde{V}^{(k)}(t,x) &= \EVtx\sBrackets{\rBrackets{\int \limits_{t}^{T}e^{-\delta \rBrackets{s - t}}\left(\hat{c}(s)\right)^{\rho}\,ds + e^{-\delta \rBrackets{T - t}}\rBrackets{X^{\hat{\pi}, \hat{c}}_T}^{\rho}}^{\frac{1 - \alpha}{\rho} - k}}.
	\end{align*}
\end{corollary}

From the first-order conditions for the supremum in \eqref{eq:KM_CT_PDE_for_V}, we obtain the candidate equilibrium controls:
\begin{align}
	\hat{\pi}(t,x) &= -\frac{\partial_x V(t, x) \lambda}{\partial_{xx} V(t, x) x \sigma^2}, \label{eq:KM_CT_candidate_for_equilibrium_pi}\\
	\hat{c}(t,x) &= \rBrackets{\frac{\partial_x V(t, x) }{\widetilde{V}^{(1)}(t,x)}}^{\frac{1}{\rho - 1}}. \label{eq:KM_CT_candidate_for_equilibrium_c}
\end{align}
At this point, we conjecture that the variables $t$ and $x$ can be separated via the ansatz
\begin{align*}
	V(t,x) & = \frac{1}{1 - \alpha} A(t) x^{1 - \alpha},\\
	\widetilde{V}^{(k)}(t,x) & =  A^{(k)}(t) x^{1 - \alpha - k\rho}, \quad k \in \mathbb{N},
\end{align*}
where $A$ and $A^{(k)}$ are functions to be determined. Using the ansatz, we rewrite the (candidate) equilibrium strategies as
\begin{align}
	\hat{\pi}(t,x) &= \dfrac{\lambda}{\alpha \sigma^{2} }, \label{eq:EquilibriumInvestmentKM}\\
	\hat{c}(t,x) &= x \left(\dfrac{A(t)}{A^{(1)}(t)}\right)^{\frac{1}{\rho-1}}. \label{eq:EquilibriumConsumptionKM}
\end{align}
Replacing \eqref{eq:EquilibriumInvestmentKM}-\eqref{eq:EquilibriumConsumptionKM} in the system of PDEs and performing straightforward calculations, we derive a system of ODEs:
\begin{equation}\label{eq:ODEsystem_CRRA-CES}
\begin{split}
0 = & \; \partial_{t}A(t) + (1-\alpha) A(t) \left(r + \dfrac{\lambda^{2}}{2\alpha\sigma^{2} } - \dfrac{\delta}{\rho}\right) - (1-\alpha)\rBrackets{1-\dfrac{1}{\rho}} (A(t))^{\frac{\rho}{\rho-1}} (A^{(1)}(t))^{-\frac{1}{\rho-1}}, \\
0 = & \; \partial_{t}A^{(k)}(t)  +(1-\alpha-k\rho) A^{(k)}(t) \left(r + \dfrac{\lambda^{2}}{2\alpha\sigma^{2}}  -  \dfrac{\delta}{\rho} - k\rho\dfrac{\lambda^{2}}{2\alpha^{2}\sigma^{2}} \right) \\
  & - (1-\alpha-k\rho)A^{(k)}(t)\left(\dfrac{A(t)}{A^{(1)}(t)}\right)^{\frac{1}{\rho-1}} \\
		& \quad + \left(\dfrac{1-\alpha}{\rho} - k \right)A^{(k+1)}(t) \left(\dfrac{A(t)}{A^{(1)}(t)}\right)^{\frac{\rho}{\rho-1}}, \quad k \in \mathbb{N}, \\
		& A(T) = 1, \\
		& A^{(k)}(T) = 1, \quad k \in \mathbb{N}.
	\end{split}
\end{equation}

\begin{remark}
	Generally, system \eqref{eq:ODEsystem_CRRA-CES} requires the solution of infinitely many equations. However, if $\frac{1-\alpha}{\rho}$ is a positive integer, let us call it $\bar{k}$, when $k = \bar{k}$ several terms cancel out and we end up with $A^{(k)}(t)\vert_{k = \bar{k}} = 1$ for any $t$. De facto, this reduces the number of equations to $\bar{k}+1$, as we would not need to consider indices $k > \bar{k}$. For instance, when $\frac{1-\alpha}{\rho} = 1$, we obtain that $A^{(1)}(t) = 1$ and $A(t)$ solves the same ODE for the time-additive CRRA utility studied in \cite{Merton1969:RES} (thus yielding, in turn, the same consumption policy).
	
\end{remark}

Before examining the consumption policy over the lifecycle, let us briefly comment on the investment strategy. It turns out that the equilibrium investment in \eqref{eq:EquilibriumInvestmentKM} coincides with the optimal investment obtained for CRRA-CES continuous-time recursive utility:
\begin{equation} \label{eq:CRRA-CES_EZpreferences_CT}
	\begin{split}
		J(t) &= \mathbb{E}_{t,x} \left[ \int_{t}^{T} m\big(c(s),J(s)\big) ds + \dfrac{1}{1-\alpha}\left(X_{T}^{\pi,c}\right)^{1-\alpha}\right],  \\
		\mbox{with} \quad  m(c,J) &:= \dfrac{1-\alpha}{\rho}\delta J \left( c^{\rho}\left( \frac{1}{(1-\alpha) J} \right)^{\frac{\rho}{1-\alpha}} - 1\right),
	\end{split}
\end{equation}
where we use the short notation $J(t) = J\left(t,x,\big(\pi, c\big)\right)$; see, for instance, \cite{KraftSeifriedSteffensen2013:FS}.  In addition, this solution is equivalent in the CRRA case (\cite{Merton1969:RES}), with the interpretative caveat that in that case, there is no wedge between the parameters underlying attitudes concerning time and risk - that is, $1-\alpha = \rho$.

A more subtle point pertains to the discussion on page 641 of \cite{Kihlstrom2009:JME}. Therein, Kihlstrom compares the solution of a two-period consumption-investment problem under the CRRA-CES specification of KM preferences with the solution under recursive utility. In that regard, he seems to state that the investment strategy under KM preferences depends on both the risk aversion and the elasticity of intertemporal substitution -- hence in conflict with the optimum under recursive utility and our equilibrium strategy \eqref{eq:EquilibriumInvestmentKM}, which depends solely on the risk aversion. However, as Kihlstrom's characterization of the maximization problem does not produce an explicit solution, how this joint dependence unfolds is \mbox{not transparent}.

\paragraph*{Numerical illustration.} To conclude this example, we display the average consumption, annuity demand,\footnote{The annuity demand can be interpreted as the inverse of the time-dependent percentage of wealth consumed by the agent.} and average wealth over time for different levels of risk aversion $\alpha$; see Figure \ref{fig:EquilibriumStrategiesPlot} (left column). Other parameters are listed in the caption. 

We compare these quantities with those obtained from a CRRA-CES recursive utility as in \eqref{eq:CRRA-CES_EZpreferences_CT}; see Figure \ref{fig:EquilibriumStrategiesPlot} (right column). Note that, in this case, the optimal consumption policy is available in closed form:
\begin{equation*}
\begin{split}
	&\hspace{3cm}  c^{*}_{\scriptscriptstyle EZ}(t,x) = \dfrac{x}{a(t)},\,\,\,\\
		 \text{where} \quad  & a(t) =  \dfrac{1}{\nu} + \left(1-\dfrac{1}{\nu}\right)e^{\nu(t-T)}, \quad \nu  = \dfrac{\delta}{1-\rho} + \left(1-\dfrac{1}{1-\rho}\right)\left(r + \dfrac{\lambda^{2}}{2\alpha\sigma^{2}}\right). 
   \end{split}
\end{equation*}
We can obtain Merton's solution by setting $1-\alpha = \rho$.\footnote{In Figure \ref{fig:EquilibriumStrategiesPlot}, with $\rho$ set equal to $-1$, we retrieve Merton's solution when $\alpha = 2$. In that case, the curves for KM and EZ (left and right column, respectively) are identical.}

\afterpage{
\begin{figure}[t!]
\begin{minipage}[t]{0.5\textwidth}
	\centering
 \centering
 \hspace{-7cm}
 \subfloat[Average consumption KM]{\includegraphics[width=1\textwidth]{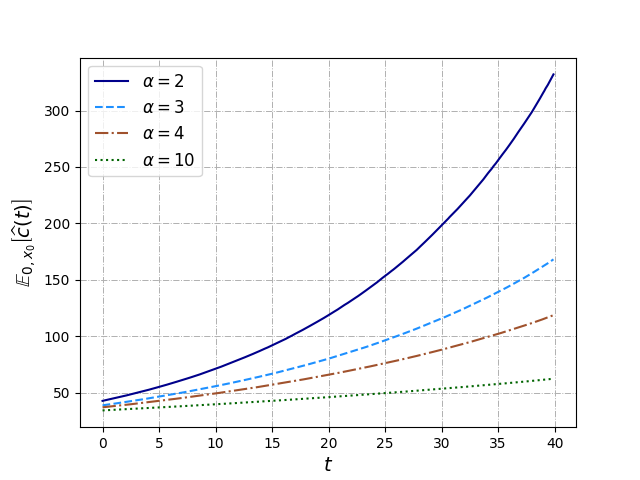}} 
			\subfloat[Average consumption EZ]{\includegraphics[width=1\textwidth]{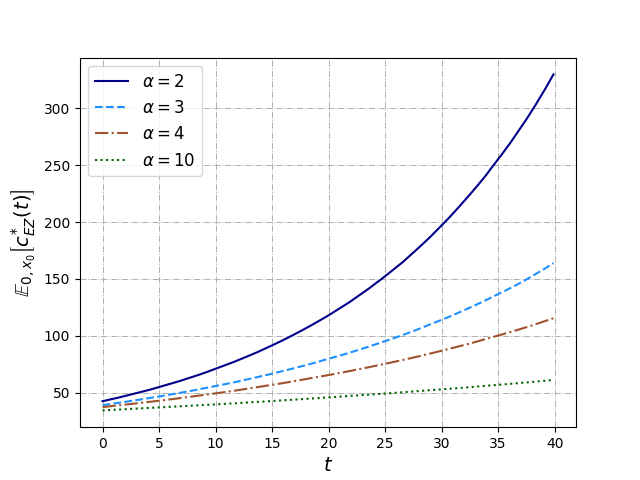}}
			\par
			\subfloat[Annuity KM]{\includegraphics[width=1\textwidth]{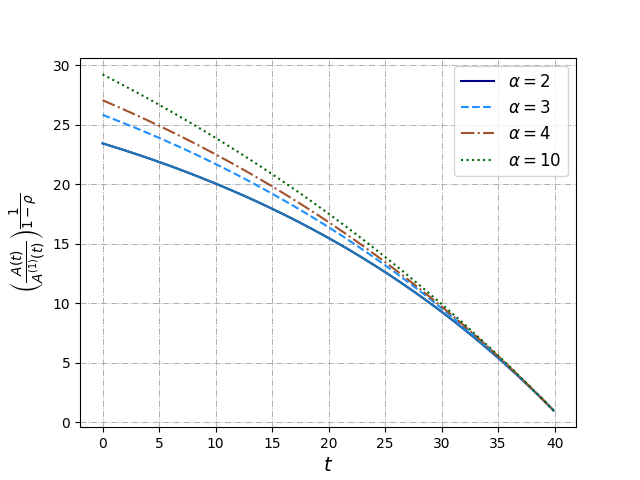}} 
			\subfloat[Annuity EZ]{\includegraphics[width=1\textwidth]{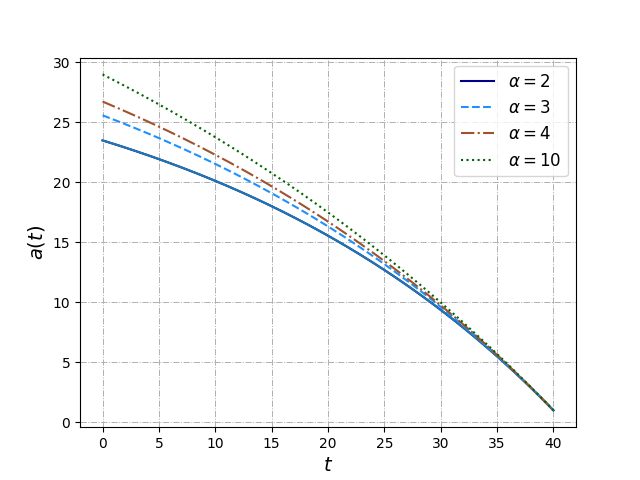}}
			\par
			\subfloat[Average wealth KM]{\includegraphics[width=1\textwidth]{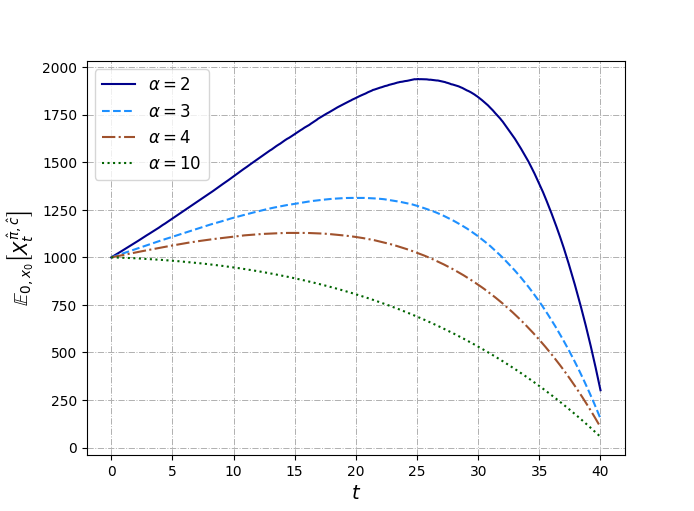}}
			\subfloat[Average wealth EZ]{\includegraphics[width=1\textwidth]{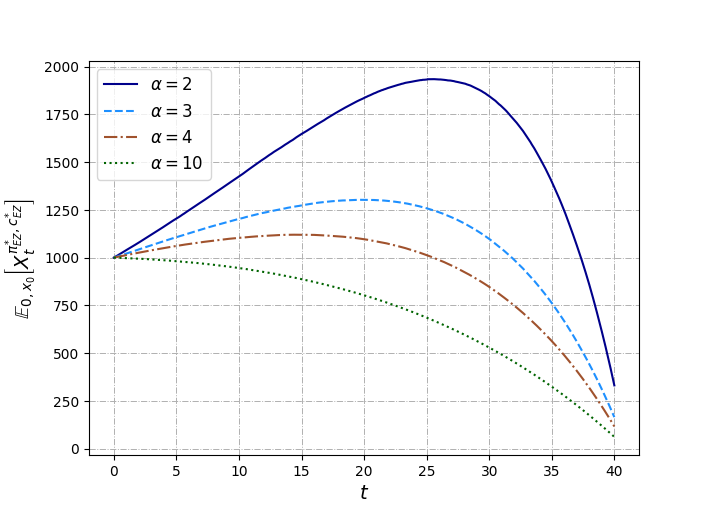}}
		\end{minipage}
		\caption{\footnotesize Average consumption, annuity demand, and average wealth over time for CRRA-CES Kihlstrom\textendash Mirman preferences \eqref{eq:CRRA-CES_RewardFunctional_KM_CT} (left column) and CRRA-CES Epstein\textendash Zin preferences \eqref{eq:CRRA-CES_EZpreferences_CT} (right column), with different levels of risk aversion $\alpha$. Market parameters: $(r,\lambda, \sigma) = (0.02, 0.07, 0.2)$. Discounting: $\delta = 0.01$. EIS: $\rho = -1$. Investment horizon: $T = 40$. Initial wealth: $x_0 = \$ 1000 $. \\
		}
		\label{fig:EquilibriumStrategiesPlot}
	\end{figure}
	\clearpage
}

Net of marginal quantitative differences, which can be seen from Table \ref{table:KMvsEZ_numericalresults}, the two models yield qualitatively similar results: average consumption increases steadily over the lifetime (more rapidly as the agent is less risk averse), both in absolute terms and in the percentage of wealth. On the other hand, wealth accumulates in the first phase. It decreases then towards the end of the time horizon - unless the agent is too risk averse, in which case the investment in the risky asset is not substantial enough to sustain an increase in wealth.

\afterpage{
	\thispagestyle{empty}
	\begin{landscape}
		\begin{table}[t!]
			\parbox{.45\linewidth}{
				\centering
				\begin{tabular}{c  c c c c c c }
					\multicolumn{7}{c}{Kihlstrom\textendash Mirman} \\
					\hline
					& & & \multicolumn{4}{c}{$t$} \\
					\cline{4-7}
					& &  \qquad & 5 & 15 & 25 & 35 \\
					\hline 
					& & $\alpha$ \qquad  & & & &\\
					\cline{3-3}
					\multirow{4}{2.cm}{$ \; \mathbb{E}_{0,x_0}\left[\hat{c}(t)\right]$} &  & 2 \qquad  & 54.95 & 91.9 & 153.39 & 256.77 \\
					& & 3 \qquad  & 46.34 & 66.89 & 96.47 & 139.76 \\
					& & 4 \qquad & 42.62 & 57.01 & 76.25 & 102.39 \\
					& & 10 \qquad & 36.81 & 42.74 & 49.63 & 57.74\\
					\hline 
					\multirow{4}{2cm}{$\left(\frac{A(t)}{A^{(1)}(t)}\right)^{\frac{1}{1-\rho}}$} &  & 2 \qquad  & 21.88 & 17.95 & 12.63 & 5.42 \\
					& & 3 \qquad  & 23.9 & 19.22 & 13.20 & 5.51 \\
					& & 4 \qquad & 24.93 & 19.83 & 13.47 & 5.55 \\
					& & 10 \qquad & 26.69 & 20.86 & 13.91 & 5.62 \\
					\hline 
					\multirow{4}{2cm}{ $\quad \mathbb{E}_{0,x_0}\left[X_t^{\hat{\pi},\hat{c}}\right]$} &  & 2 \qquad  & 1202.33 & 1649.7 & 1936.63 & 1390.53 \\
					& & 3 \qquad  & 1107.96 & 1285.54 & 1273.87 & 770.31 \\
					& & 4 \qquad & 1062.32 & 1130.76 & 1027.28 & 568.63 \\
					& & 10 \qquad & 982.69 & 891.65 & 690.09 & 324.44 
				\end{tabular}
			}
			\hspace{1cm}
			\parbox{.45\linewidth}{
				\centering
				\begin{tabular}{c  c c c c c c }
					\multicolumn{7}{c}{Epstein\textendash Zin} \\
					\hline
					& & & \multicolumn{4}{c}{$t$} \\
					\cline{4-7}
					& &  \qquad & 5 & 15 & 25 & 35 \\
					\hline 
					& & $\alpha$ \qquad  & & & &\\
					\cline{3-3}
					\multirow{4}{2cm}{$ \; \mathbb{E}_{0,x_0}\left[c_{EZ}^{*}(t)\right]$} &  & 2 \qquad  & 54.87 & 91.78 & 153.19 & 256.43 \\
					& & 3 \qquad  & 46.71 & 66.96 & 95.79 & 137.43 \\
					& & 4 \qquad & 43.04 & 57.11 & 75.69 & 100.45 \\
					& & 10 \qquad & 37.06 & 42.77 & 49.33 & 56.98\\
					\hline 
					\multirow{4}{2cm}{$\hspace{1cm}a(t)$} &  & 2 \qquad  & 21.92 & 17.99 & 12.69 & 5.49 \\
					& & 3 \qquad  & 23.66 & 19.08 & 13.17 & 5.58 \\
					& & 4 \qquad & 24.62 & 19.65 & 13.43 & 5.62 \\
					& & 10 \qquad & 26.47 & 20.75 & 13.9 & 5.69 \\
					\hline 
					\multirow{4}{2cm}{ $\quad \mathbb{E}_{0,x_0}\left[X_{t}^{\pi^{*}_{EZ},c^{*}_{EZ}}\right]$} &  & 2 \qquad  & 1202.69 & 1651.81 & 1943.64 & 1410.09 \\
					& & 3 \qquad  & 1105.6 & 1277.55 & 1262.1 & 766.77 \\
					& & 4 \qquad & 1059.65 & 1122.45 & 1015.88 & 564.53 \\
					& & 10 \qquad & 981.04 & 887.67 & 685.88 & 324.49
				\end{tabular}
			}
			\caption{Selected values from Figure \ref{fig:EquilibriumStrategiesPlot} on the average consumption, annuity demand, and average wealth over time for CRRA-CES Kihlstrom\textendash Mirman preferences \eqref{eq:CRRA-CES_RewardFunctional_KM_CT} (left panel) and CRRA-CES Epstein\textendash Zin preferences \eqref{eq:CRRA-CES_EZpreferences_CT} (right panel). \\ 
				Averages are computed via Monte Carlo simulation of 100000 paths. To avoid differences due to random number generators, all experiments have been initialized with the same seed.}
			\label{table:KMvsEZ_numericalresults}
		\end{table}
		\raisebox{0cm}{\makebox[\linewidth]{\thepage}}
		
	\end{landscape}
	\clearpage
}

That being said, the stark similarity between the results for KM preferences and recursive utility suggests that, at least contextually to the common assumption of constant elasticity of intertemporal substitution and constant relative risk aversion, the two models align on a fundamental level. 

We can strengthen this point further by considering a deterministic setting. Let $\bar{c} := (\bar{c}_t)_{t \in \mathcal{T}}$ be a deterministic consumption stream and evaluate
\begin{equation*}
	U_{t}^{\bar{c}} := \dfrac{1}{1-\alpha}\left( \int_{t}^{T} e^{-\delta(s-t)}(\bar{c}_s)^{\rho} ds + e^{-\delta (T-t)} (X^{\bar{c}}_{T})^{\rho} \right)^{\frac{1-\alpha}{\rho}}
\end{equation*}
as the total KM utility over $\bar{c}$. Differentiating with respect to $t$, we find that $U^{\bar{c}}$ follows the dynamics
\begin{equation*}
	\begin{split}
		& dU_{t}^{\bar{c}} = \left(-\dfrac{1}{\rho}(\bar{c}_t)^{\rho}\big((1-\alpha) U_{t}^{\bar{c}}\,\big)^{1-\frac{\rho}{1-\alpha}}  + \delta \dfrac{1-\alpha}{\rho}U_{t}^{\bar{c}}\right) dt,  \\
		& U_{T}^{\bar{c}}= \dfrac{1}{1-\alpha} \left(X^{\bar{c}}_T\right)^{1-\alpha}.
	\end{split}
\end{equation*}
Therefore, we can write
\begin{equation} \label{eq:deterministicKM}
	\begin{split}
		U_{t}^{\bar{c}} & = \int_{t}^{T}\left(\,\dfrac{1}{\rho} (\bar{c}(s))^{\rho}\big((1-\alpha) U_{s}^{\bar{c}}\,\big)^{1-\frac{\rho}{1-\alpha}}  -\delta\dfrac{1-\alpha}{\rho}U_{s}^{\bar{c}}  \,\right) ds  + \dfrac{1}{1-\alpha}(X^{\bar{c}}_T)^{1-\alpha}.
	\end{split}
\end{equation}
Thus, in the absence of risk, $U^{\bar{c}}$ has a similar recursive representation of the EZ utility \eqref{eq:CRRA-CES_EZpreferences_CT}.\footnote{We thank Ninna Reitzel Heegaard (née Jensen) for directing our attention to \eqref{eq:deterministicKM}. Based on this relation, she obtained in unpublished calculations a result analogous to Corollary \ref{cor:CT_KM_with_cur_t_dependence_EHJB_system_infinite} during her PhD studies.} In uncertain environments, the difference between KM and EZ preferences is due to how the expectation of future outcomes is formed (or, in other words, the position of the conditional expectation operator).

\section{Conclusions and outlook} \label{sec:Conclusions}

The long-standing approach of \cite{KihlstromMirman1974:JET, KihlstromMirman1981:RES} to separate between risk aversion and substitution across goods has been scarcely applied in intertemporal choice models due to issues of time-inconsistency. Our analysis provided a template to tackle such time-inconsistent problems in continuous-time Markovian environments from an equilibrium point of view. 

One can picture several directions for further exploration. In particular, it should be possible to consider variations of the preferences considered herein that include state-dependence (for instance, in the spirit of \cite{BjorkMurgociZhou2014:MF}, with a state-dependent risk aversion) or robust formulations that account for the uncertainty concerning the underlying prior (as in \cite{GilboaSchmeidler1989:JME}, \cite{KlibanoffMarinacciMukerji2005:Econometrica}) or concerning risk attitudes (\cite{DesmettreSteffensen2023:MF}).

Throughout the paper, we assumed what \cite{Kihlstrom2009:JME} called a forward-looking reward functional, because the aggregation of utilities only concerns immediate and future actions. Another interesting avenue would be to study the objectives of the form
\begin{equation*} 
	\mathbb{E}_{t}\left[ \Phi \left(\int_{s \geq 0} H(s,c(s),t) ds \right) \right], \quad t \geq 0,
\end{equation*}
that is, the risk assessment also comprehends past consumptions (it is both backward- and forward-looking). We speculate that this is related to the discussion around precommitted strategies.

\begin{appendix}
\section{Proofs}
\label{app:proofs_for_general_framework}

\subsection{Proof of Lemma \ref{lem:CT_KM_with_tau_recursion_for_f}}\label{subsec:proof_lemma_CT_KM_with_tau_recursion_for_f}
Starting with the definition of $f^{\cu}$ in \eqref{eq:KM_CT_with_tau_f_probabilistic_representation}, we have the following equalities:
\begin{align*}
	f^{\cu}(t,x,z,\tau) & \\
	& \hspace{-2cm}= \EVtx\sBrackets{\Phi\rBrackets{\int \limits_{t}^{T}H(s, X_s^{\cu}, \cu(s),\tau)\,ds + G \rBrackets{ X^{\cu}_T,\tau} + z}} \\
	& \hspace{-2cm} = \EVtx\left[\mathbb{E}_{t+h, X_{t+h}^{\cu}} \Biggl[ \Phi\left(\,\,\intthT H(s, X_s^{\cu}, \cu(s),\tau)\,ds \right. \right.\\
	& \hspace{3.cm} \left. \left. + G \rBrackets{ X^{\cu}_T,\tau} + \inttth H(s, X_s^{\cu}, \cu(s),\tau)\,ds + z \right) \Bigg] \right] \\
	& \hspace{-2cm} = \EVtx\left[f^{\cu}\rBrackets{t+h,  X_{t+h}^{\cu}, \inttth H(s, X_s^{\cu}, \cu(s),\tau)\,ds + z, \tau}\right],
\end{align*}
which proves \eqref{eq:CT_KM_with_tau_f_cu_recursion_rule}. From this, \eqref{eq:CT_KM_with_tau_f_cu_recursion_terminal_condition} follows readily after:
\begin{align*}
	f^{\cu}(T, x, z, \tau) & = \mathbb{E}_{T,x}\sBrackets{\Phi\rBrackets{\int \limits_{T}^{T}H(s, X_s^{\cu}, \cu(s), \tau)\,ds + G \rBrackets{ X^{\cu}_T, \tau} + z } } \\
	& = \mathbb{E}_{T,x} \sBrackets{ \Phi\rBrackets{ G \rBrackets{ X^{\cu}_T, \tau} + z } } = \Phi\rBrackets{ G \rBrackets{ x, \tau } + z }.
\end{align*}
Finally, we prove \eqref{eq:CT_KM_with_tau_PDE_f_cu}. Rewriting \eqref{eq:CT_KM_with_tau_f_cu_recursion_rule} for a fixed $\tau$, we have
\begin{align*}
	0 = \EVtx\sBrackets{f^{\cu|\tau}\rBrackets{t + h, X^{\cu}_{t+h},  \inttth H^{|\tau}(s, X_s^{\cu}, \cu(s)) \, ds + z}} - f^{\cu|\tau}(t, x, z).
\end{align*}
Dividing by $h>0$ and taking the limit as $h \downarrow 0$ gives that
\begin{align*}
	0 &= \lim_{h \downarrow 0} \frac{1}{h}\rBrackets{\EVtx\sBrackets{f^{\cu|\tau}\rBrackets{t + h, X^{\cu}_{t+h},  \inttth H^{|\tau}(s, X_s^{\cu}, \cu(s)) \, ds + z}} - f^{\cu|\tau}(t, x, z)}\\
	& = \lim_{h \downarrow 0} \frac{1}{h}\mathcal{D}_h f^{\cu|\tau}(t,x,z),
\end{align*}
where we defined 
\begin{equation}\label{eq:def_D_h_f_cu_with_tau}
	\mathcal{D}_{h}f^{\cu|\tau}(t,x, z) := \EV_{t,x, z}\sBrackets{f^{\cu|\tau}\left(t + h, X^{\cu}_{t+h}, Z^{\cu|\tau}_{t+h}\right)} - f^{\cu|\tau}(t,x, z),
\end{equation}
with the dynamics of $\rBrackets{Z^{\cu|\tau}_s}_{s \in [t, T]}$ given in \eqref{eq:SDE_Z}. Applying It\^{o}'s lemma to the two-dimensional process $\rBrackets{X^{\cu}_t, Z^{\cu|\tau}_t}_{\tin}$ and the function $f^{\cu|\tau}(t, x, z)$, we have
\begin{align*}
	f^{\cu|\tau}&\rBrackets{ t + h, X^{\cu}_{t+h}, Z^{\cu| \tau}_{t+h}}  - f^{\cu|\tau}(t,x,z) \\
	& = \inttth \partial_t f^{\cu|\tau}\rBrackets{s, X^{\cu}_{s}, Z^{\cu| \tau}_{s}}\,ds  + \inttth \partial_z f^{\cu|\tau}\rBrackets{s, X^{\cu}_{s}, Z^{\cu| \tau}_{s}}\, H^{|\tau}(s,X^{\cu}_{s}, \cu(s))\,ds \\
	& \quad + \inttth \partial_x f^{\cu|\tau}\rBrackets{s, X^{\cu}_{s}, Z^{\cu| \tau}_{s}}\, \Big(\mu\rBrackets{s, X^{\cu}_{s}, \cu(s)}\,ds + \sigma\rBrackets{s, X^{\cu}_{s}, \cu(s)}dW_s\Big)
	\\
	& \quad + \frac{1}{2}\inttth \partial_{xx} f^{\cu|\tau}\rBrackets{s, X^{\cu}_{s}, Z^{\cu| \tau}_{s}}\, \Big(\sigma\rBrackets{s, X^{\cu}_{s}, \cu(s)}\Big)^2\,ds\\
	&= \inttth \mathcal{D}^{\cu} f^{\cu|\tau}\rBrackets{s, X^{\cu}_{s}, Z^{\cu| \tau}_{s}}\,ds + \inttth \partial_x f^{\cu|\tau}\rBrackets{s, X^{\cu}_{s}, Z^{\cu| \tau}_{s}} \sigma\rBrackets{s, X^{\cu}_{s}, \cu(s)}\,dW_s,
\end{align*}
where in the first equality we use the definition of $X^{\cu}$ and $Z^{\cu}$, and the fact that $d\langle Z^{\cu| \tau}, Z^{\cu| \tau} \rangle_t = 0$ and \mbox{$d \langle X^{\cu}, Z^{\cu| \tau} \rangle _t = 0$}, and in the second equality we use the definition of $\mathcal{D}^{\cu}$. Taking the expectation and using the fact that the stochastic integral is a martingale with expectation $0$, we obtain
\begin{align}
	\EVtxz \Bigl[ f^{\cu|\tau}&\rBrackets{ t + h, X^{\cu}_{t+h}, Z^{\cu|\tau}_{t+h} }  - f^{\cu|\tau}(t,x,z) \Bigr] \notag\\
	& = \EVtxz \sBrackets{\inttth \mathcal{D}^{\cu} f^{\cu|\tau}\rBrackets{s, X^{\cu}_{s}, Z^{\cu| \tau}_{s}}\, ds}.\label{eq:A_h_f_tau_integral_representation}
\end{align}
Finally, dividing \eqref{eq:def_D_h_f_cu_with_tau} by $h$, taking the limit as $h\downarrow 0$, and using \eqref{eq:A_h_f_tau_integral_representation}, we derive
\begin{align*}
	\lim_{h \downarrow 0} \frac{1}{h} \mathcal{D}_h f^{\cu|\tau}(t,x,z) & = \lim_{h \downarrow 0} \frac{1}{h} \; \EVtxz \sBrackets{f^{\cu|\tau}\rBrackets{ t + h, X^{\cu}_{t+h}, Z^{\cu|\tau}_{t+h} }  - f^{\cu|\tau}(t,x,z)} \\
	& = \lim_{h \downarrow 0} \frac{1}{h} \; \EVtxz \sBrackets{\inttth \mathcal{D}^{\cu} f^{\cu|\tau}\rBrackets{s, X^{\cu}_{s}, Z^{\cu| \tau}_{s}}\, ds} \\
	& = \EVtxz \sBrackets{\lim_{h \downarrow 0} \frac{1}{h} \inttth \mathcal{D}^{\cu} f^{\cu|\tau}\rBrackets{s, X^{\cu}_{s}, Z^{\cu| \tau}_{s}}\, ds} \\
	& = \EVtxz \sBrackets{\lim_{h \downarrow 0} \frac{1}{h} \mathcal{D}^{\cu} f^{\cu|\tau}\rBrackets{ \eta, X^{\cu}_{\eta}, Z^{\cu|\tau}_{\eta}}\, h}\\
	&= \EVtxz \sBrackets{\mathcal{D}^{\cu} f^{\cu|\tau}(t, x, z)} = \mathcal{D}^{\cu}f^{\cu|\tau}(t,x,z),
\end{align*}
where in the third equality we use the dominated convergence theorem (guaranteed by the assumption $f^{\cu} \in \mathcal{L}^2(X^{\cu})$) to pass the limit inside the expectation, and in the fourth equality we apply the mean value theorem $ \omega$-wise and $\eta(\omega) \in [t, t+ h]$ for $\omega \in \Omega$.\qed

\subsection{Proof of Theorem \ref{th:verification_theorem_with_tau}}\label{subsec:proof_verification_theorem_with_tau}

\textit{Proof of (R1).} By (C1), $\ecu$ realizes the supremum in \eqref{eq:CT_KM_with_tau_EHJB_PDE_V_t_x_z} and is admissible. Thus, we consider a stochastic process $\rBrackets{Z^{\ecu|\tau}_s}_{s \in [t, T]}$ that satisfies \eqref{eq:SDE_Z} for $\cu = \ecu$. By the regularity assumption on $f$ in (C3), we can apply It\^{o}'s lemma to $f^{|\tau}(t, X^{\ecu}_t, Z^{\ecu|\tau}_t)$ on $[t,T]$ and, analogously to $f^{\cu|\tau}\rBrackets{ t + h, X^{\cu}_{t+h}, Z^{\cu| \tau}_{t+h}}  - f^{\cu|\tau}(t,x,z)$ in the proof of Lemma \ref{lem:CT_KM_with_tau_recursion_for_f}, get  that
\begin{align*}
	f^{|\tau}&\rBrackets{T, X^{\ecu}_{T}, Z^{\ecu|\tau}_{T} }  - f^{|\tau}(t, X^{\ecu}_{t}, Z^{\ecu|\tau}_{t}) \\
	& = \inttT \mathcal{D}^{\ecu} f^{|\tau}(s, X^{\ecu}_{s}, Z^{\ecu|\tau}_{s}) \,ds +  \inttT \partial_x f^{|\tau}\rBrackets{s, X^{\ecu}_{s}, Z^{\ecu|\tau}_{s}} \sigma\rBrackets{s, X^{\ecu}_{s}, \ecu(s)}\,dW_s.
\end{align*}
By (C2), $f$ satisfies \eqref{eq:CT_KM_with_tau_EHJB_PDE_f_T_x_z_tau}, which implies that $\int_t^{T}\mathcal{D}^{\ecu} f^{|\tau}(s, X^{\ecu}_{s}, Z^{\ecu|\tau}_{s}) \,ds= 0$. And by (C4), $f \in \mathcal{L}^2(X^{\ecu})$, thus $ \int_t^{T} \partial_x f^{|\tau}\rBrackets{s, X^{\ecu}_{s}, Z^{\ecu|\tau}_{s}} \sigma\rBrackets{s, X^{\ecu}_{s}, \ecu(s)}\,dW_s$ is a martingale and its expectation is equal to $0$. Therefore, we obtain
\begin{equation*}
	\mathbb{E}_{t, x, z}\sBrackets{f^{|\tau}\rBrackets{T, X^{\ecu}_{T}, Z^{\ecu|\tau}_{T} }  - f^{|\tau}\rBrackets{t, X^{\ecu}_{t}, Z^{\ecu|\tau}_{t}}} = 0.
\end{equation*}
Using the linearity of the expectation operator and the above equality, we get
\begin{align}
	\mathbb{E}_{t, x, z}\sBrackets{f^{|\tau}(t, X^{\ecu}_{t}, Z^{\ecu|\tau}_{t})} &= \mathbb{E}_{t, x, z}\sBrackets{f^{|\tau}\rBrackets{T, X^{\ecu}_{T}, Z^{\ecu|\tau}_{T} } }  = \mathbb{E}_{t, x, z}\sBrackets{f\rBrackets{T, X^{\ecu}_{T}, Z^{\ecu|\tau}_{T}, \tau} } \notag \\
	& \stackrel{\eqref{eq:CT_KM_with_tau_EHJB_PDE_f_T_x_z_tau}}{=} \mathbb{E}_{t, x, z}\sBrackets{\Phi\rBrackets{G\rBrackets{X^{\ecu}_{T},\tau} + Z^{\ecu|\tau}_{T}}}. \label{eq:f_t_x_z_tau_representation_2nd_last_step}
\end{align}
From the SDE of $Z^{\ecu|\tau}$, we get $Z^{\ecu|\tau}_T = \int_t^{T} H^{|\tau}(s, X_s^{\ecu}, \ecu(s))\,ds + Z^{\ecu|\tau}_t$. Plugging it into the right-hand side of \eqref{eq:f_t_x_z_tau_representation_2nd_last_step} and using that $Z^{\ecu|\tau}_t = z$, we obtain \eqref{eq:KM_CT_with_tau_f_probabilistic_representation}, which proves (R1).

\textit{Proof of (R2).} We now prove that $V(t,x) = J(t,x,\ecu)$ follows from the assumptions of the theorem and from what we have shown in (R1). By (C2), $V$ solves \eqref{eq:CT_KM_with_tau_EHJB_PDE_V_t_x_z}. By (C1), $\ecu$ realizes the supremum in \eqref{eq:CT_KM_with_tau_EHJB_PDE_V_t_x_z} and is admissible. Thus:
\begin{equation}\label{eq:KM_CT_with_tau_EHJB_with_V_at_ecu}
	0  = \mathcal{D}^{\ecu}V(t,x) + \partial_{z} f(t, x, 0, t) H(t, x, \ecu(t), t) - \partial_{\tau} f(t, x, 0, t).
\end{equation}
Since $V$ satisfies (C3), we can apply It\^{o}'s lemma to $V(t, X^{\ecu}_t)$ on $[t, T]$ and obtain:
\begin{align*}
	V&\rBrackets{T, X^{\ecu}_{T}}  - V(t, X^{\ecu}_{t}) = \inttT \mathcal{D}^{\ecu} V(s, X^{\ecu}_{s}) \,ds +  \inttT \partial_x V\rBrackets{s, X^{\ecu}_{s}} \sigma\rBrackets{s, X^{\ecu}_{s}, \ecu(s)}\,dW_s.
\end{align*}
Rearranging the terms and taking expectations yields
\begin{align*}
	\EVtx\sBrackets{V(t, X^{\ecu}_{t})} = & \;\EVtx\sBrackets{V\rBrackets{T, X^{\ecu}_{T}}}  - \EVtx\sBrackets{\inttT \mathcal{D}^{\ecu} V(s, X^{\ecu}_{s}) \,ds} \\
	& -  \EVtx\sBrackets{\inttT \partial_x V\rBrackets{s, X^{\ecu}_{s}} \sigma\rBrackets{s, X^{\ecu}_{s}, \ecu(s)}\,dW_s}.
\end{align*}
For the first term in the right-hand side of the above equality, we use that $V(T,x) = \Phi\rBrackets{G\rBrackets{x,T}}$ by (C2). For the second term, we use that $\mathcal{D}^{\ecu}V(t,x) = \partial_{\tau} f(t, x, 0, t) - \partial_z f(t, x, 0, t) H(t, x, \ecu(t), t)$ by \eqref{eq:KM_CT_with_tau_EHJB_with_V_at_ecu}. For the third term, we use the fact that $V \in \mathcal{L}^2(X^{\ecu})$ by (C4), which implies again that the stochastic integral is a martingale with expectation $0$. Therefore, we derive that
\begin{align}
	\underbrace{\EVtx\sBrackets{V(t, X^{\ecu}_{t})}}_{=V(t,x)} & = \EVtx\sBrackets{\Phi\rBrackets{G\rBrackets{X^{\ecu}_{T}, T}}}  - \EVtx\sBrackets{\inttT \partial_{\tau} f(s, X^{\ecu}_{s}, 0, s)\,ds} \notag \\
	& \quad +  \EVtx\sBrackets{\inttT \partial_z f(s, X^{\ecu}_{s}, 0, s) H(s, X^{\ecu}_{s}, \ecu(s), s) \,ds}. \label{eq:KM_CT_with_tau_V_t_x_representation_intermediate_step_1}
\end{align}
Next, define $g(t,x) := f(t, x, 0, t)$. Since $f$ is sufficiently smooth by (C3), $g$ is sufficiently smooth too. Applying It\^{o}'s lemma to $g(t,X^{\ecu}_t)$ on  $[t, T]$, we have
\begin{align*}
	g & (T, X^{\ecu}_T) - g(t,X^{\ecu}_t) = \inttT \partial_t g(s,X^{\ecu}_s)\,ds + \inttT \partial_x g(s, X^{\ecu}_s)\,dX^{\ecu}_s \\
	& \quad  + \frac{1}{2} \inttT \partial_{xx} g(s,X^{\ecu}_s)\,d\langle X^{\ecu}, X^{\ecu} \rangle_s \\
	& = \inttT \partial_t g(s,X^{\ecu}_s)\,ds + \inttT \partial_x g(s, X^{\ecu}_s) \mu(s, X^{\ecu}_s, \ecu(s))\,ds \\
	& \quad + \frac{1}{2} \inttT \partial_{xx} g(s,X^{\ecu}_s) \rBrackets{\sigma\rBrackets{s, X^{\ecu}_{s}, \ecu(s)}}^2\,ds  + \inttT \partial_x g(s, X^{\ecu}_s) \sigma(s, X^{\ecu}_s, \ecu(s))\,dW_s.
\end{align*}
Observing that the partial derivatives of $g(t,x)$ are equivalent to
\begin{align*}
	\partial_t g(t,x) &= \partial_t f(t, x, 0, t) + \partial_{\tau} f(t, x, 0, t),\\
	\partial_x g(t,x) & = \partial_x f(t, x, 0, t),\\
	\partial_{xx} g(t,x) & = \partial_{xx} f(t, x, 0, t),
\end{align*}
we get
\begin{align*}
	& f(T, X^{\ecu}_T, 0, T)  
    = f(t,X^{\ecu}_t, 0, t) + \inttT \partial_t f(s,X^{\ecu}_s, 0, s)\,ds + \inttT \partial_{\tau} f(s,X^{\ecu}_s, 0, s)\,ds    \\
	& \quad + \inttT \partial_x f(s, X^{\ecu}_s, 0, s) \mu(s, X^{\ecu}_s, \ecu(s))\,ds + \frac{1}{2} \inttT \partial_{xx} f(s, X^{\ecu}_s, 0, s) \rBrackets{\sigma\rBrackets{s, X^{\ecu}_{s}, \ecu(s)}}^2\,ds  \\
	& \quad +  \inttT \partial_z f(s, X^{\ecu}_s, 0, s) H(s, X^{\ecu}_s, \ecu(s), s)\,ds + \inttT \partial_x f(s, X^{\ecu}_s, 0, s) \sigma(s, X^{\ecu}_s, \ecu(s))\,dW_s \\
	& \quad - \inttT \partial_z f(s, X^{\ecu}_s, 0, s) H(s, X^{\ecu}_s, \ecu(s), s) \,ds \\
	& = f(t,X^{\ecu}_t, 0, t) + \inttT \mathcal{D}^{\ecu}f^{|s}(s, X^{\ecu}_s, 0)\,ds + \inttT \partial_x f(s, X^{\ecu}_s, 0, s) \sigma(s, X^{\ecu}_s, \ecu(s))\,dW_s \\
	&\quad + \inttT \partial_{\tau}f(s,X^{\ecu}_s, 0, s)\,ds - \inttT \partial_z f(s, X^{\ecu}_s, 0, s)H(s, X^{\ecu}_s, \ecu(s), s)\,ds\\
	& = f(t,X^{\ecu}_t, 0, t)  + \inttT \partial_x f(s, X^{\ecu}_s, 0, s) \sigma(s, X^{\ecu}_s, \ecu(s))\,dW_s \\
	& \quad + \inttT \partial_{\tau} f(s,X^{\ecu}_s, 0, s)\,ds - \inttT \partial_z f(s, X^{\ecu}_s, 0, s)H(s, X^{\ecu}_s, \ecu(s), s)\,ds,
\end{align*}
where in the first equality we add and subtract the term $\int_t^{T} \partial_z f(s, X^{\ecu}_s, 0, s) \linebreak \times H(s, X^{\ecu}_s, \ecu(s), s)\,ds$, in the second equality we use the definition of $\mathcal{D}^{\ecu}f^{|\tau}$, and in the third equality that $\mathcal{D}^{\ecu}f^{|t}(t,x,0) = 0$, as by (C2) $f$ satisfies \eqref{eq:CT_KM_with_tau_EHJB_PDE_f_t_x_z} for any $(t,x, z, \tau)$. Rearranging the terms and applying the expectation operator, we have
\begin{align}
	\EVtx\Bigg[ & \inttT \partial_z  f (s, X^{\ecu}_s, 0, s)  H(s, X^{\ecu}_s, \ecu(s), s)\,ds\Biggr]  \notag\\
	& = \EVtx\sBrackets{f(t,X^{\ecu}_t, 0, t)} + \EVtx\sBrackets{\inttT \partial_x f(s, X^{\ecu}_s, 0, s) \sigma(s, X^{\ecu}_s,\ecu(s))\,dW_s} \notag\\
	& \quad - \EVtx\sBrackets{f(T, X^{\ecu}_T, 0, T)} + \EVtx\sBrackets{\inttT \partial_{\tau}f(s,X^{\ecu}_s, 0, s)\,ds}\notag\\
	& = f(t,x, 0, t) - \EVtx\sBrackets{f(T, X^{\ecu}_T, 0, T)}  + \EVtx\sBrackets{\inttT \partial_{\tau} f(s,X^{\ecu}_s, 0, s)\,ds} \notag \\
	& = \mathbb{E}_{t,x} \sBrackets{\Phi\rBrackets{\int \limits_{t}^{T}H(s, X_s^{\ecu}, \ecu(s), t)\,ds + G \rBrackets{ X^{\ecu}_T, t} }} - \EVtx\sBrackets{\Phi\rBrackets{G\rBrackets{X^{\ecu}_{T}, T}}} \notag \\
	& \quad + \EVtx\sBrackets{\inttT \partial_{\tau} f(s,X^{\ecu}_s, 0, s)\,ds}, \label{eq:KM_CT_with_tau_V_t_x_representation_intermediate_step_2}
\end{align}
where in the second equality we use (C2) regarding the terminal condition satisfied by $f$ and the probabilistic representation \eqref{eq:KM_CT_with_tau_f_probabilistic_representation} of $f$, which was proven earlier in (R1). Inserting \eqref{eq:KM_CT_with_tau_V_t_x_representation_intermediate_step_2} into \eqref{eq:KM_CT_with_tau_V_t_x_representation_intermediate_step_1}, we get
\begin{align}
	V(t,x) & = \EVtx\sBrackets{\Phi\rBrackets{G\rBrackets{X^{\ecu}_{T}, T}}} - \EVtx\sBrackets{\inttT \partial_{\tau} f(s, X^{\ecu}_{s}, 0, s)\,ds} \notag \\
	& \quad +  \mathbb{E}_{t,x} \sBrackets{\Phi\rBrackets{\int \limits_{t}^{T}H(s, X_s^{\ecu}, \ecu(s), t)\,ds + G \rBrackets{ X^{\ecu}_T, t} } } - \EVtx\sBrackets{\Phi\rBrackets{G\rBrackets{X^{\ecu}_{T}, T}}} \notag \\
	&\quad  + \EVtx\sBrackets{\inttT \partial_{\tau} f(s, X^{\ecu}_{s}, 0, s)\,ds} \notag \\
	& = \mathbb{E}_{t,x} \sBrackets{\Phi\rBrackets{\int \limits_{t}^{T}H(s, X_s^{\ecu}, \ecu(s), t)\,ds + G \rBrackets{ X^{\ecu}_T, t} }} = J(t, x, \ecu),\notag
\end{align}
thus proving (R2).

\textit{Proof of (R3).} Next, we show that $\ecu$ is an equilibrium control as per Definition \ref{def:equilibrium_u}. Fix an arbitrary point $(t, x) \in [0, T) \times \mathcal{X}$. Choose an admissible control $\cu \in \ASet$ and a real number $h \in [0, T - t]$. Define a new control $\cuh$ as in \eqref{eq:perturbed_equilibrium_control}. We will show that \eqref{eq:ecu_def_property} holds, i.e.:
\begin{equation*}
	\liminf_{h \downarrow 0} \frac{J(t,x, \ecu) - J(t,x, \cuh)}{h} \geq 0.
\end{equation*}
Let us derive a recursive representation for $J(t,x, \cuh)$. We will use the auxiliary functions $f^{\cuh}$ and $f^{\ecu}$, where the latter is equal to $f$ due to \eqref{eq:KM_CT_with_tau_f_probabilistic_representation} proven in (R1). We have
\begin{align*}
	J(t, x, \cuh) & \stackrel{\eqref{eq:KM_CT_with_tau_f_probabilistic_representation_z_0_tau_t}}{=} f^{\cuh}(t, x, 0, t) \\
	& \stackrel{\eqref{eq:CT_KM_with_tau_f_cu_recursion_rule}}{=} \EVtx\sBrackets{f^{\cuh}\rBrackets{t + h, X^{\cuh}_{t+h},  \inttth H(s, X_s^{\cuh}, \cuh(s), t) \, ds, t}}.
\end{align*}
Adding and subtracting $J(t + h, X^{\cuh}_{t+h}, \cuh)$ in the right-hand side of the previous equation, and using \eqref{eq:KM_CT_with_tau_f_probabilistic_representation_z_0_tau_t}, we obtain
\begin{align}
	J(t, x, \cuh) & = \EVtx\sBrackets{f^{\cuh}\rBrackets{t + h, X^{\cuh}_{t+h},  \inttth H(s, X_s^{\cuh}, \cuh(s), t) \, ds, t}}  \notag\\
	& \quad + J(t + h, X^{\cuh}_{t+h}, \cuh) - J(t + h, X^{\cuh}_{t+h}, \cuh) \notag\\
	& = J(t + h, X^{\cuh}_{t+h}, \cuh) \notag\\
	& \quad + \EVtx\sBrackets{f^{\cuh}\rBrackets{t + h, X^{\cuh}_{t+h},  \inttth H(s, X_s^{\cuh}, \cuh(s), t) \, ds, t}} \notag \\
	& \quad - f^{\cuh}(t + h, X^{\cuh}_{t+h}, 0, t + h). \label{eq:CT_KM_with_tau_verifying_ecu_intermediate_step_1}
\end{align}
By the definition of $\cuh$, we also have the following equalities:
\begin{align*}
	& J(t + h, X^{\cuh}_{t+h}, \cuh) = V(t + h, X^{\cu}_{t+h}),\\
	& f^{\cuh}\rBrackets{t + h, X^{\cuh}_{t+h},  \inttth H(s, X_s^{\cuh}, \cuh(s), t) \, ds, t } \\
	& \hspace{0.2cm}\qquad = f^{\ecu}\rBrackets{t + h, X^{\cu}_{t+h},  \inttth H(s, X_s^{\cu}, \cu(s), t) \, ds, t }, \\
	& f^{\cuh}\left(t + h, X^{\cuh}_{t+h}, t+h\right) = f^{\ecu}\left(t + h, X^{\cu}_{t+h}, 0, t + h\right). \\
	& 
\end{align*}
Plugging the results into \eqref{eq:CT_KM_with_tau_verifying_ecu_intermediate_step_1} and applying the expectation operator, we get
\begin{align}
	J(t, x, \cuh) &= \EVtx\sBrackets{V(t + h, X^{\cu}_{t+h})}  \notag \\
	& \quad + \EVtx\sBrackets{f\rBrackets{t + h, X^{\cu}_{t+h},  \inttth H(s, X_s^{\cu}, \cu(s), t) \, ds, t}} \notag \\
	&\quad - \EVtx\sBrackets{f(t + h, X^{\cu}_{t+h}, 0, t + h)} .\label{eq:CT_KM_with_tau_verifying_ecu_intermediate_step_2}
\end{align}
Therefore, using \eqref{eq:CT_KM_with_tau_verifying_ecu_intermediate_step_2}, we have
\begin{align}
	J(t, x, \ecu) &- J(t,x, \cuh) \notag\\
	& = V(t, x) - J(t,x, \cuh) \notag \\
	& =: -\mathcal{D}_h^{\cu} V(t,x) - \mathcal{D}_{h}^{\cu}f^{|t+h, X^{\cu}_{t+h}, t}(z)|_{z = 0} + \mathcal{D}_h f^{|t + h, X^{\cu}_{t+h}, 0}(\tau)|_{\tau = t}, \label{eq:KM_CT_with_tau_verifying_ecu_intermediate_step_3}
\end{align}
where we specify the following notation:
\begin{align}
	\mathcal{D}_{h}^{\cu}V(t,x) & := \EV_{t,x}\sBrackets{V(t + h, X^{\cu}_{t+h})} - V(t,x), \label{eq:CT_KM_with_tau_def_D_h_V}\\
	\mathcal{D}_{h}^{\cu}f^{|t+h, X^{\cu}_{t+h}, t}(z) & := \EV_{t,x, z}\sBrackets{f^{|t + h, X^{\cu}_{t+h}, t}\left(Z^{\cu|t}_{t + h}\right)} - \EV_{t,x, z}\sBrackets{f^{|t + h, X^{\cu}_{t+h}, t}(z)}, \label{eq:CT_KM_with_tau_def_D_h_f_z}\\
	\mathcal{D}_h f^{|t + h, X^{\cu}_{t+h}, 0}(\tau) &:= \EVtx\sBrackets{f^{|t + h, X^{\cu}_{t+h}, 0}(\tau + h)} - \EVtx\sBrackets{f^{|t + h, X^{\cu}_{t+h}, 0}(\tau)} ,\label{eq:CT_KM_with_tau_def_D_h_f_tau}
\end{align}
for the process $\rBrackets{Z^{\cu|t}_{s}}_{s \in [t, t + h]}$ defined in \eqref{eq:SDE_Z} with $\tau = t$ fixed. Dividing both sides of \eqref{eq:KM_CT_with_tau_verifying_ecu_intermediate_step_3} by $h$, taking the limit as $h \downarrow 0$ and using the linearity of the limit operator yields that
\begin{equation} \label{eq:KM_CT_with_tau_verifying_ecu_intermediate_step_4}
	\begin{aligned}
		\liminf_{h \downarrow 0}& \frac{J(t, x, \ecu) - J(t,x, \cuh)}{h} = -\liminf_{h \downarrow 0}\frac{1}{h}\mathcal{D}_h^{\cu} V(t,x) \\
		& - \liminf_{h \downarrow 0}\frac{1}{h}\mathcal{D}_{h}^{\cu}f^{|t+h, X^{\cu}_{t+h}, t}(z)|_{z = 0} + \liminf_{h \downarrow 0}\frac{1}{h}\mathcal{D}_h f^{|t + h, X^{\cu}_{t+h}, 0}(\tau)|_{\tau = t}.
	\end{aligned}
\end{equation}
As for the first term in \eqref{eq:KM_CT_with_tau_verifying_ecu_intermediate_step_4}, due to $V \in L^2(X^{\cu})$ by (C4), we have
\begin{equation}\label{eq:CT_KM_with_tau_A_h_V_t_x_convergence}
	\liminf_{h \downarrow 0} \frac{1}{h}\mathcal{D}_{h}^{\cu}V(t,x) = \mathcal{D}^{\cu}V(t,x).
\end{equation}
To simplify the second term in \eqref{eq:KM_CT_with_tau_verifying_ecu_intermediate_step_4}, we apply It\^{o}'s lemma to the function $g(z):=f^{|t+h, X^{\cu}_{t+h}, t}(z)$, the process $\rBrackets{Z_s^{\cu|t}}_{s \in [t, t + h]}$ for $Z_t^{\cu} = z = 0$, that is:
\begin{equation}\label{eq:SDE_Z_tau_fixed_to_t_z_to_0}
	d Z_s^{\cu|t} = H^{|t}(s, X^{\cu}_s, \cu(s))ds,\quad Z_t^{\cu|t} = 0.
\end{equation}
We obtain that
\begin{align}
	g\left(Z_{t+h}^{\cu|t}\right) &= g\left(Z_{t}^{\cu|t}\right) + \inttth \partial_{z} g\left(Z_{s}^{\cu|t}\right)\,d Z_{s}^{\cu|t} + \frac{1}{2} \inttth \partial_{zz} g\left(Z_{s}^{\cu|t}\right) \,d\langle Z^{\cu|t}, Z^{\cu\vert t} \rangle_s \notag\\
	& = g\left(Z_{t}^{\cu|t}\right) + \inttth \partial_{z} g\left(Z_{s}^{\cu|t}\right)H^{\vert t}(s, X^{\cu}_s, \cu(s))\,ds,\label{eq:g_z_ito_lemma}
\end{align}
which follows from \eqref{eq:SDE_Z_tau_fixed_to_t_z_to_0}. 
Then, from the definition of $\mathcal{D}_{h}^{\cu}f^{|t+h, X^{\cu}_{t+h}, t}(z)$ in \eqref{eq:CT_KM_with_tau_def_D_h_f_z}, and using \eqref{eq:g_z_ito_lemma}, we compute
\begin{align}
	\liminf_{h \downarrow 0} \frac{1}{h} \notag & \mathcal{D}_{h}^{\cu}f^{|t+h, X^{\cu}_{t+h}, t}(0) \notag \\
	= & \liminf_{h \downarrow 0} \frac{1}{h} \; \EV_{t,x, 0} \sBrackets{g\left(Z_{t+h}^{\cu|t}\right) - g\left(Z_{t}^{\cu|t}\right)}  \notag \\
	= &  \liminf_{h \downarrow 0} \frac{1}{h}  \; \EV_{t,x, 0} \sBrackets{g\left(Z_{t}^{\cu|t}\right) + \inttth \partial_{z} g\left(Z_{s}^{\cu|t}\right)H^{|t}(s, X^{\cu}_s, \cu(s))\,ds - g\left(Z_{t}^{\cu|t}\right)} \notag \\
	= & \; \liminf_{h \downarrow 0} \frac{1}{h} \; \EV_{t,x, 0} \sBrackets{\partial_{z} g\left(Z_{\eta}^{\cu|t}\right)H^{|t}\left(\eta, X^{\cu}_{\eta}, \cu(\eta)\right)h}  \notag \\
	= & \;  \EV_{t,x, 0} \sBrackets{\liminf_{h \downarrow 0} \rBrackets{ \partial_{z} f^{\vert \, t + h, X^{\cu}_{t+h}, t}\left(Z_{\eta}^{\cu|t}\right) H^{|t}\left(\eta, X^{\cu}_{\eta}, \cu(\eta)\right)} }  \notag \\
	= & \; \EV_{t,x, 0} \sBrackets{\partial_{z} f^{\vert \, t, X^{\cu}_{t}, t}\left(Z_{t}^{\cu|t}\right)H^{\vert \, t}\left(t, X^{\cu}_{t}, \cu(t)\right)} \notag  \\
	= &  \partial_{z} f^{\vert \, t, x, t}(0)H^{|t}(t, x, \cu(t)) \notag \\
	= & \; \partial_{z} f(t, x, 0, t)H(t, x, \cu(t), t), \label{eq:CT_KM_with_tau_A_h_f_z_convergence} 
\end{align}
where in the third equality we apply the mean value theorem $\omega$-wise, for $ \omega \in \Omega$, and $\eta(\omega) \in [t, t+ h]$, in the fourth equality we use (C5) to apply the dominated convergence theorem to interchange expectation and limit operators, and in the last steps we use the fact that $Z^{\cu|t}_t = 0$ and the definition of $\partial_{z} f^{\vert t, x, t}$. Thus, the second term in \eqref{eq:KM_CT_with_tau_verifying_ecu_intermediate_step_3} converges to $\partial_{z} f(t, x, 0, t)H(t, x, \cu(t), t)$. Similarly, for the convergence of the third term in \eqref{eq:KM_CT_with_tau_verifying_ecu_intermediate_step_3}, starting from the definition of $\mathcal{D}_h f^{|t + h, X^{\cu}_{t+h}, 0}(\tau)$ in \eqref{eq:CT_KM_with_tau_def_D_h_f_tau}, we derive
\begin{align}
	\liminf_{h \downarrow 0}\frac{1}{h}& \mathcal{D}_h f^{|t + h, X^{\cu}_{t+h}, 0}(\tau)|_{\tau = t}  \notag \\
	&=  \liminf_{h \downarrow 0}\frac{1}{h}\bigg(\EVtx\sBrackets{f^{|t + h, X^{\cu}_{t+h}, 0}(t + h)} - \EVtx\sBrackets{f^{|t + h, X^{\cu}_{t+h}, 0}(t)}\bigg) \notag\\
	& =\liminf_{h \downarrow 0}{\EVtx\sBrackets{\frac{1}{h}\rBrackets{f^{|t + h, X^{\cu}_{t+h}, 0}(t + h) - f^{|t + h, X^{\cu}_{t+h}, 0}(t)}}} \notag \\
	& = \liminf_{h \downarrow 0}\EVtx\Biggl[\frac{1}{h}\biggl(f^{|t + h, X^{\cu}_{t+h}, 0}(t) + ( t + h - t)\partial_{\tau}f^{|t + h, X^{\cu}_{t + h}, 0}(t) \notag \\
	&\qquad \qquad \qquad + \dfrac{1}{2}(t + h - t)^2\partial_{\tau \tau}f^{|t + h, X^{\cu}_{t+h}, 0}(\iota)  - f^{|t + h, X^{\cu}_{t+h}, 0}(t) \biggr)\Biggr] \notag \\
	& = \liminf_{h \downarrow 0}\EVtx\Biggl[\frac{1}{h}\biggl(h \, \partial_{\tau}f^{|t + h, X^{\cu}_{t + h}, 0}(t)   + \dfrac{1}{2}  h^2 \partial_{\tau \tau }f^{|t + h, X^{\cu}_{t+h}, 0}(\iota)  \biggr)\Biggr] \notag \\
	& = \liminf_{h \downarrow 0}\EVtx\Biggl[\partial_{\tau}f^{|t + h, X^{\cu}_{t + h}, 0}(t) \Biggr] + \liminf_{h \downarrow 0}\EVtx\Biggl[\dfrac{1}{2}h \,\partial_{\tau \tau }f^{|t + h, X^{\cu}_{t+h}, 0}(\iota)   \Biggr] \notag \\
	& =\EVtx\Biggl[\liminf_{h \downarrow 0}\partial_{\tau}f^{|t + h, X^{\cu}_{t + h}, 0}(t) \Biggr] + \EVtx\Biggl[\liminf_{h \downarrow 0}\dfrac{1}{2}h \,\partial_{\tau \tau }f^{|t + h, X^{\cu}_{t+h}, 0}(\iota)  \Biggr] \notag \\
	& = \EVtx\Biggl[\partial_{\tau}f^{|t, X^{\cu}_{t}, 0}(t) \Biggr] + \EVtx\Biggl[0 \cdot \partial_{\tau \tau}f^{|t, X^{\cu}_{t}, 0}(t)  \Biggr] \notag \\
	& = \partial_{\tau} f(t, x, 0, t),\label{eq:CT_KM_with_tau_A_h_f_tau_convergence}
\end{align}
where in the third equality we use for each $\omega \in \Omega$ the first-order Taylor expansion of $f^{|t + h, X^{\cu}_{t+h}(\omega), 0}(\tau)$ around the point $\tau = t$, with $\iota(\omega) \in [t, t+ h]$ being a random intermediate value, in the fourth equality we apply again the dominated convergence theorem (using (C6)), and in the last steps the continuity of $\partial_{\tau} f$ and $\partial_{\tau \tau} f$ with respect to $t$ and $x$. 

Finally, inserting \eqref{eq:CT_KM_with_tau_A_h_V_t_x_convergence}, \eqref{eq:CT_KM_with_tau_A_h_f_z_convergence} and \eqref{eq:CT_KM_with_tau_A_h_f_tau_convergence} into \eqref{eq:KM_CT_with_tau_verifying_ecu_intermediate_step_4}, we obtain
\begin{align*}
	\liminf_{h \downarrow 0} & \frac{J(t, x, \ecu) - J(t,x, \cuh)}{h} \\
	& = -\rBrackets{\mathcal{D}^{\cu} V(t,x) + \partial_{z} f(t, x, 0, t)H(t, x, \cu(t,x), t) - \partial_{\tau} f(t, x, 0, t)} \geq 0,
\end{align*}
due to \eqref{eq:CT_KM_with_tau_EHJB_PDE_V_t_x_z} as per (C1). Therefore, $\ecu$ is an equilibrium control,
which proves (R3).

\textit{Proof of (R4).} We conclude that $V(t, x)$ is indeed the equilibrium value function, i.e., $\widehat{V}(t, x) = V(t, x)$, for $V(t, x) = J(t, x, \ecu)$ by (R2) and $\ecu$ is an equilibrium control by (R3). This proves (R4) and completes the proof of the verification theorem. \qed

\subsection{Proof of Corollary \ref{cor:CT_KM_with_cur_t_dependence_EHJB_system_infinite}}
\label{subsec:ProofCorollaryReformulationHJBSystemCRRA}
For readability, we recall the system of PDEs in \eqref{eq:EHJBsystem_CRRA}:
\begin{align} 
		0 =  & \; \sup_{(\pi, c) \in \mathcal{A}(t,x)}\Big\{\partial_{t} V(t,x) + \partial_x V(t, x) (x(r + \pi \lambda) - c) + \dfrac{1}{2}\partial_{xx} V(t, x) \sigma^2 \pi^2 x^2  \Big. \notag \\
		& \hspace{2cm} \Big. + \partial_{z} f(t, x, 0, t) c^\rho- \partial_{\tau} f(t, x, 0, t) \Big\}, \label{eq:EHJBsystem_CRRAproof} \\
		0 =  & \; \mathcal{D}^{\hat{\pi},\hat{c}  }f^{\vert \tau}(t,x,z), \notag \\
		& V(T,x) = \dfrac{1}{1 - \alpha} x^{1 - \alpha},\quad \notag \\
		& f(T,x,z,\tau) = \dfrac{1}{1 - \alpha}\rBrackets{e^{-\delta \rBrackets{T - \tau}}  x^{\rho} + z}^{\frac{1 - \alpha}{\rho}}. \notag
\end{align}
Let us compute the term $\partial_{\tau} f(t, x, 0, t)$ in the first equation:
\begin{align*}
	\partial_{\tau} f(t, x, 0, t) & =  \rBrackets{\partial_{\tau} f(t,x,0,\tau)}_{\tau = t} \\
	&\hspace{-1.5cm} \stackrel{}{=} \delta\frac{1 - \alpha}{\rho} \EVtx \sBrackets{ \frac{1}{1 - \alpha}\rBrackets{\int \limits_{t}^{T} e^{-\delta (s - t) } (\hat{c}(s))^{\,\rho}\,ds + e^{-\delta (T - t) }\rBrackets{X^{\hat{\pi}, \hat{c}}_T}^{\rho}}^{\frac{1 - \alpha}{\rho}} } \\
	& \hspace{-1.5cm} = \delta \frac{1 - \alpha}{\rho} V(t, x),
\end{align*}
where in the first equality we use the definition of $f(t,x,z,\tau)$, and in the last equality the definition of $V(t, x) = f(t, x, 0, t)$. Next, we compute the term $\partial_{z} f(t, x, 0, t)$: 
\begin{equation*}
	\partial_{z} f(t, x, 0, t) =\EVtx\sBrackets{\frac{1}{\rho}\rBrackets{\int \limits_{t}^{T} e^{-\delta \rBrackets{s - t}} (\hat{c}(s))^{\,\rho} \,ds + e^{-\delta \rBrackets{T - t}}\rBrackets{X^{\hat{\pi}, \hat{c}}_T}^{\rho} }^{\frac{1 - \alpha}{\rho}-1}}.
\end{equation*}
Defining $\widetilde{V}^{(1)}(t,x) := \rho \partial_{z} f(t, x, 0, t)$, the first equation in \eqref{eq:EHJBsystem_CRRAproof} becomes
\begin{equation}
	0 = \sup_{(\pi,c) \in \mathcal{A}(t,x)}\left\{\mathcal{D}^{\pi,c }V(t,x) + \dfrac{1}{\rho}\widetilde{V}^{(1)}(t,x) c^{\rho} - \delta \frac{1 - \alpha}{\rho} V(t, x)  \right\},
\end{equation}
with terminal condition for $\widetilde{V}^{(1)}(t,x)$ given by $\widetilde{V}^{(1)}(T,x) =  x^{1 - \alpha - \rho}.$ Now let
\begin{equation*}
	f^{(1)}(t,x,z,\tau) := \EVtx\sBrackets{\frac{1}{\rho}\rBrackets{\int \limits_{t}^{T} e^{-\delta \rBrackets{s - \tau}} \left(\hat{c}(s)\right)^{\rho}\,ds + e^{-\delta \rBrackets{T - \tau}}\rBrackets{X^{\hat{\pi}, \hat{c}}_T}^{\rho} + z }^{\frac{1 - \alpha}{\rho}-1}}.
\end{equation*}
Analogously to the proof of Lemma \ref{lem:CT_KM_with_tau_recursion_for_f}, we can show that for any $h \in [0, T - t],$ $\tau \in [0, t]$,
\begin{align*}
	f^{(1)}(t,x,z,\tau) 
	& \stackrel{}{=} \EVtx\sBrackets{f^{(1)}\rBrackets{t+h,  X_{t+h}^{\hat{\pi},\hat{c}}, \inttth e^{-\delta \rBrackets{s - \tau}} \left(\hat{c}(s)\right)^{\rho}\,ds + z, \tau}}.
\end{align*}
Thus, $f^{(1)}(t,x,z,\tau)$ satisfies the following PDE:
\begin{eqnarray}
		 && 0 = \mathcal{D}^{\hat{\pi},\hat{c}}f^{(1) \, \vert \, \tau}(t,x,z) \notag  \\
		 && \;\;\, = \partial_{t}f^{(1)}(t,x,z, \tau) + \partial_x f^{(1)}(t,x,z, \tau)\mu\big(t, x, (\hat{\pi}(t,x),\hat{c}(t,x)) \big) \label{eq:PDE_f1_CRRAproof} \\
		&& \qquad + \frac{1}{2}\partial_{xx}f^{(1)}(t,x,z, \tau) \rBrackets{\sigma\big(t, x, (\hat{\pi}(t,x),\hat{c}(t,x)) \big)}^2 + \partial_z f^{(1)}(t,x, z, \tau)(\hat{c}(t,x))^{\,\rho}. \notag
\end{eqnarray}
Inserting $z = 0$ and $\tau = t$, and using the definition of $\widetilde{V}^{(1)}(t,x)$, we obtain
\begin{align*}
	0 & = \partial_t\widetilde{V}^{(1)}(t,x) + \partial_x\widetilde{V}^{(1)}(t,x)\mu\big(t, x, (\hat{\pi}(t,x),\hat{c}(t,x))\big) \\  
	& \qquad + \frac{1}{2}\partial_{xx}\widetilde{V}^{(1)}(t,x) \rBrackets{\sigma\big(t, x, (\hat{\pi}(t,x),\hat{c}(t,x))\big)}^2 - \rho \partial_{\tau} f(t, x, 0, t) \\
	& \qquad + \rho \partial_z f^{(1)}(t,x,0, t)(\hat{c}(t,x))^{\,\rho}\\
	& = \mathcal{D}^{\hat{\pi},\hat{c}} \, \widetilde{V}^{(1)}(t,x) - \rho \partial_{\tau}f^{(1)}(t, x, 0, t) + \rho \partial_z f^{(1)}(t,x,0, t)(\hat{c}(t,x))^{\,\rho}.
\end{align*}
Furthermore, by similar calculations to those performed above, notice that
$\allowdisplaybreaks \\ \partial_\tau f^{(1)}(t,x,0, t) = \frac{\delta}{\rho} \rBrackets{\frac{1 - \alpha}{\rho} - 1}\widetilde{V}^{(1)}(t,x),$
and that
\begin{align*}
	\partial_z f^{(1)}(t,x,0, t) & = \partial_z \EVtx\sBrackets{\frac{1}{\rho}\rBrackets{\int \limits_{t}^{T} e^{-\delta \rBrackets{s - t}} \left(\hat{c}(s)\right)^{\rho} \,ds + e^{-\delta \rBrackets{T - t}}\rBrackets{X^{\hat{\pi}, \hat{c}}_T}^{\rho} }^{\frac{1 - \alpha}{\rho}-1}}\\
	& \hspace{-2cm} = \EVtx\sBrackets{\frac{1}{\rho} \rBrackets{\frac{1 - \alpha}{\rho} - 1}\rBrackets{\int \limits_{t}^{T} e^{-\delta \rBrackets{s - t}} \left(\hat{c}(s)\right)^{\rho} \,ds + e^{-\delta \rBrackets{T - t}}\rBrackets{X^{\hat{\pi}, \hat{c}}_T}^{\rho} }^{\frac{1 - \alpha}{\rho}-2}}.
\end{align*}
Therefore, letting $\widetilde{V}^{(2)}(t,x) := \rho\rBrackets{\frac{1 - \alpha}{\rho} - 1}^{-1} \partial_z f^{(1)}(t,x,0, t),$
we rewrite \eqref{eq:PDE_f1_CRRAproof} as
\begin{equation*}
	0 = \mathcal{D}^{\hat{\pi},\hat{c}} \,\widetilde{V}^{(1)}(t,x) - \delta\rBrackets{\frac{1 - \alpha}{\rho} - 1}\widetilde{V}^{(1)}(t,x) +  \rBrackets{\frac{1 - \alpha}{\rho} - 1}\widetilde{V}^{(2)}(t,x)(\hat{c}(t,x))^{\,\rho}.
\end{equation*}
Repeating the same procedure for all $ k \geq 2$, we obtain the PDE for a generic $\widetilde{V}^{(k)}(t,x)$:
\begin{equation*}
	\begin{split}
		& 0  =  \mathcal{D}^{\hat{\pi},\hat{c}}\widetilde{V}^{(k)}(t,x) - \delta\rBrackets{\dfrac{1 - \alpha}{\rho} - k}\widetilde{V}^{(k)}(t,x) +  \rBrackets{\dfrac{1 - \alpha}{\rho} - k}\widetilde{V}^{(k+1)}(t,x) (\hat{c}(t,x))^{\,\rho}, \\
		& \widetilde{V}^{(k)}(T,x)  =  x^{1-\alpha - k\rho}.
	\end{split}
\end{equation*}
Putting everything together, the claim of the corollary follows. \qed

\section{Special cases}

In this appendix, we first specialize the theory of Section \ref{sec:general_theory} to the case where the utility functions $H$ and $G$ do not depend on the current time. Then, we provide details on the consumption-investment problem studied in Section \ref{sec:CRRA-CES_utility_exponential_discounting} for the case of $\alpha = 1$ (unit RRA).

\subsection{No dependence on current time} \label{app:NoDependenceOnCurrentTime}

Consider the setup introduced in Section \ref{subsec:setup}, with the agent's reward functional given by
\begin{equation}\label{eq:KM_CT_without_tau_reward_functional}
	J(t, x, \cu) = \EVtx \sBrackets{\Phi\rBrackets{\int \limits_{t}^{T}H(s, X_s^{\cu}, \bm{u}(s, X_s^{\cu}))\,ds + G \rBrackets{X_T^{\cu}} } }.
\end{equation}
In the spirit of Corollary \ref{cor:CT_KM_with_cur_t_dependence_EHJB_system_infinite}, we aim to show that we can rewrite the extended HJB system in an alternative form. 
\begin{proposition}\label{prop:CT_KM_without_tau_EHJB_system_infinite} Assume that the following conditions are satisfied:
	\begin{enumerate}
		\item[(C7)] An admissible equilibrium control $\ecu$ exists and realizes the supremum in $$0  = \sup_{u \in \AMap(t,x)}\left\{\mathcal{D}^{u}V(t,x) + \partial_{z} f(t, x, 0) H(t, x, u) \right\}.$$
		\item[(C8)] $\Phi \in \textgoth{C}^{\infty}\left( \R \right)$.
		\item[(C9)] For $\Phi_{[r]}(x)$, $r \in \mathbb{N}$, denoting the $r$-th order derivative of $\Phi(x)$, and by convention, $\Phi_{[0]}(x) := \Phi(x)$, it holds that
		\begin{equation*}
			\begin{aligned}
				\sup_{n \in \N \cup \{ 0 \}} \; \Biggl| \; & \EVtx\Biggl[{\Phi_{[n]}\Biggl({\int \limits_{t}^{T}H(s, X_s^{\ecu}, \ecu(s))\,ds + G \rBrackets{ X^{\ecu}_T}}\Biggr)}\Biggr] \Biggr| < \infty.
			\end{aligned}
		\end{equation*}

	\end{enumerate}
	Then, the extended HJB system reads:
	\begin{align*}
		0 & = \sup_{u \in \mathcal{A}(t, x)}\left\{\mathcal{D}^{u}V(t,x) + V^{(1)}(t, x) H(t, x, u) \right\}, \\
		0 & =  \mathcal{D}^{\ecu}V^{(k)}(t,x) + V^{(k+1)}(t, x) H(t, x, \ecu(t, x)), \quad k \in \mathbb{N}, \\
		V(T,x) & = \Phi\rBrackets{G\rBrackets{x}},  \\ 
		V^{(k)}(T,x) & = \Phi_{[k]}\rBrackets{G\rBrackets{x}},\quad \quad k \in \mathbb{N}.
	\end{align*}
	Furthermore, $V^{(k)}$ has the probabilistic representation
	\begin{equation*}
		V^{(k)}(t,x) = \EVtx\sBrackets{\Phi_{[k]}\rBrackets{\int \limits_{t}^{T}H\left(s, X_s^{\ecu}, \ecu(s)\right)\,ds + G \rBrackets{ X^{\ecu}_T}}}, \quad k \in \mathbb{N}.
	\end{equation*}
\end{proposition}

\begin{proof}
	By (C7), we know from Theorem \ref{th:verification_theorem_with_tau} that the extended HJB system is given by
	\begin{align*}
		0 & = \sup_{u \in \AMap(t,x)}\left\{\mathcal{D}^{u}V(t,x) + \partial_{z} f(t, x, 0) H(t, x, u) \right\}, \\
		0 & =  \mathcal{D}^{\ecu}f(t,x,z), 
	\end{align*}
	where
	\begin{equation*}
		f(t, x, z) = \mathbb{E}_{t,x} \sBrackets{\Phi\rBrackets{\int \limits_{t}^{T}H\left(s, X_s^{\ecu}, \ecu(s)\right)\,ds + G \rBrackets{ X^{\ecu}_T} + z}}.
	\end{equation*}
	Note that, in this case, the auxiliary function does not depend on $\tau$. 
	
	Let $V^{(1)}(t,x) := \partial_{z} f(t,x,0)$ and observe that
	\begin{align*}
		V^{(1)}(t,x) & =  \rBrackets{\partial_z\EVtx\sBrackets{\Phi\rBrackets{\int \limits_{t}^{T}H\left(s, X_s^{\ecu}, \ecu(s)\right)\,ds + G \rBrackets{ X^{\ecu}_T} + z}}}_{z = 0}\\
		& =\rBrackets{\EVtx\sBrackets{{\partial_z} \Phi\rBrackets{\int \limits_{t}^{T}H\left(s, X_s^{\ecu}, \ecu(s)\right)\,ds + G \rBrackets{ X^{\ecu}_T} + z}}}_{z = 0}\\
		& = \EVtx\sBrackets{\Phi_{[1]}\rBrackets{\int \limits_{t}^{T}H\left(s, X_s^{\ecu}, \ecu(s)\right)\,ds + G \rBrackets{ X^{\ecu}_T}}},
	\end{align*}
	where, in the second equality, we use the dominated convergence theorem, which is justified by (C9). Define a new auxiliary function $f^{(1)}(t,x,z)$ that equals $V^{(1)}(t,x)$ for $z=0$:
	\begin{equation*}
		f^{(1)}(t,x,z) := \EVtx\sBrackets{\Phi_{[1]}\rBrackets{\int \limits_{t}^{T}H\left(s, X_s^{\ecu}, \ecu(s)\right)\,ds + G \rBrackets{ X^{\ecu}_T}+ z} }.
	\end{equation*}
	Analogously to the derivation of the recursion for $f^{\cu}$ in the proof of Lemma \ref{lem:CT_KM_with_tau_recursion_for_f}, we have
	\begin{align*}
		f^{(1)}(t,x,z) & =\EVtx\sBrackets{\Phi_{[1]}\rBrackets{\int \limits_{t}^{T}H\left(s, X_s^{\ecu}, \ecu(s)\right)\,ds + G \rBrackets{ X^{\ecu}_T} + z}} \\
		&  = \EVtx\Bigg[\mathbb{E}_{t+h, X_{t+h}^{\ecu}} \Bigg[\Phi_{[1]}\Bigg(\,\,\intthT H\left(s, X_s^{\ecu}, \ecu(s)\right)\,ds + G \rBrackets{ X^{\ecu}_T} \Bigg. \Bigg. \Bigg.\\
		&  \Bigg. \Bigg. \Bigg. \hspace{5cm}+ \inttth H\left(s, X_s^{\ecu}, \ecu(s)\right)\,ds + z \Bigg) \Bigg] \Bigg] \\
		& = \EVtx\sBrackets{f^{(1)}\rBrackets{t+h,  X_{t+h}^{\ecu}, \inttth H\left(s, X_s^{\ecu}, \ecu(s)\right)\,ds + z}}.
	\end{align*}
	Thus, we have
	\begin{align*}
		0 = \EVtx\sBrackets{f^{(1)}\rBrackets{t+h,  X_{t+h}^{\ecu}, \inttth H\left(s, X_s^{\ecu}, \ecu(s)\right)\,ds + z}} - f^{(1)}(t,x, z),
	\end{align*}
	for any $h \in [0, T - t]$. Dividing both sides of the above equality by $h \downarrow 0$, one can show as in Lemma \ref{lem:CT_KM_with_tau_recursion_for_f} that $\mathcal{D}^{\ecu}f^{(1)}(t,x,z) = 0$, or equivalently, that 
	\begin{align*}
		0 &= \partial_{t} f^{(1)}(t,x,z) + \partial_{x}f^{(1)}(t,x,z)\mu(t,x,\ecu(t,x)) + \partial_{z}f^{(1)}(t,x,z)H(t,x,\ecu(t,x))\\
		& \quad + \frac{1}{2} \partial_{xx}f^{(1)}(t,x,z) \rBrackets{\sigma(t, x,\ecu(t,x))}^2.
	\end{align*}
	Inserting $z = 0$ into the above PDE, we obtain
	\begin{align*}
		0 & = \partial_{t}f^{(1)}(t,x,0) + \partial_{x} f^{(1)}(t,x,0)\mu(t,x,\ecu(t,x)) + \partial_{z} f^{(1)}(t,x,0)H(t,x,\ecu(t,x))\\
		& \quad + \partial_{xx}\frac{1}{2}f^{(1)}(t,x,0) \rBrackets{\sigma(t, x,\ecu(t,x))}^2\\
		& = \partial_{t}V^{(1)}(t,x) + \partial_x V^{(1)}(t,x)\mu(t,x,\ecu(t,x)) + \frac{1}{2} \partial_{xx}V^{(1)}(t,x) \rBrackets{\sigma(t, x,\ecu(t,x))}^2 \\
		& \quad + \partial_{z}f^{(1)}(t,x,0)H(t,x,\ecu(t,x))\\
		& =\mathcal{D}^{\ecu}V^{(1)}(t,x) +  \underbrace{\partial_{z}f^{(1)}(t,x,0)}_{=:V^{(2)}(t,x)}H(t,x,\ecu(t,x)),
	\end{align*}
	where we use that $V^{(1)}(t,x) = f^{(1)}(t,x, 0)$. The terminal condition is given by
	\begin{equation*}
		V^{(1)}(T,x) = \Phi_{[1]}(G(x)).
	\end{equation*}
	Repeating the process for $V^{(2)}(t,x)$, we obtain that
	\begin{align*}
		& 0 = \mathcal{D}^{\ecu}V^{(2)}(t,x) +  \underbrace{\partial_z f^{(2)}(t,x,0)}_{=:V^{(3)}(t,x)}H(t,X^{\ecu}_t,\ecu(t,x)),\\
		& V^{(2)}(T,x)  = \Phi_{[2]}(G(x)).
	\end{align*}
	By (C8) and (C9), we can repeat the same process for $V^{(k)}(t,x) := (\partial_{z})^{k}f(t,x,0)$ for any $k \in \mathbb{N}$. The claim of the proposition follows. 
\end{proof}

\begin{remark}
    The decision-making problem with the reward functional in \eqref{eq:KM_CT_without_tau_reward_functional} is time-inconsistent in the original state space if the function $\Phi$ is nonlinear. Therefore, the extended system of HJB equations (or its equivalent representation in terms of a system of infinitely many equations as per Corollary \ref{cor:CT_KM_with_cur_t_dependence_EHJB_system_infinite}) is needed to characterize an equilibrium strategy and equilibrium value function. But, if we enlarge the state space from $X^{\cu}$ to $\rBrackets{X^{\cu}, Z^{\cu}}$ with $dZ^{\cu}_t = H(t, X^{\cu}_t, \cu(t))dt$, the decision-making problem becomes time-consistent\footnote{We thank Jianfeng Zhang for directing our attention to this fact.} in the enlarged state space and the respective value function (which depends on $t,x, z$) satisfies the standard HJB PDE.
    
    However, for the reward functional \eqref{eq:generalized_KM_CT_reward_func_with_t_dependence}, the enlargement of the state space does not resolve the issue of time-inconsistency due to the dependence of $H$ and $G$ on the current time $t$. 
\end{remark}

 \subsection{Unit RRA} \label{subec:Example_CRRA_unitaryRRA}

Consider the same market setting utilized in Section \ref{sec:CRRA-CES_utility_exponential_discounting}. Recalling that $ \lim_{\alpha \to 1} \frac{1}{1-\alpha}x^{1-\alpha} = \log(x)$, we set
\begin{equation*}
	\begin{split}
		\Phi(x) & = \log(x), \\
		H(s,x,(\pi,c),t) & = e^{-\delta(s-t)}c^{\rho}, \\
		G(x,t) & = e^{-\delta(T-t)}x^{\rho},
	\end{split}
\end{equation*}
In this case, the reward functional becomes
\begin{equation}\label{eq:KM_CT_with_cur_t_dependence_reward_functional_alpha=1}
	\begin{aligned}
		J(t, x, (\pi, c)) =  \mathbb{E}_{t,x}\sBrackets{   \dfrac{1}{\rho}\log \rBrackets{\int \limits_{t}^{T} e^{-\delta (s-t) } \left(c(s)\right)^{\rho}\,ds + e^{-\delta (T-t) }\rBrackets{X^{\pi, c}_T}^{\rho}}},
	\end{aligned}
\end{equation}
with the equilibrium value function defined as $\widehat{V}(t,x) := J(t, x, (\hat{\pi}, \hat{c}))$.

By Proposition \ref{prop:CT_KM_without_tau_EHJB_system_infinite}, the extended HJB system is given by
\begin{equation}
	\begin{split} \label{eq:EHJBsystem_CRRA_unitaryriskaversion}
		\displaystyle
		0 = & \; \sup_{(\pi,c)\in \mathbb{R}^{2}} \Big\{\partial_{t}V(t,x) + \partial_x V(t, x) (x(r + \pi \lambda) - c) + \dfrac{1}{2}\partial_{xx}V(t, x) \sigma^2 \pi^2 x^2 \Big. \\
		&      \hspace{2cm} \Big. + \partial_{z} f(t, x, 0,t) c^\rho  - \partial_{\tau} f(t, x, 0, t) \Big\}, \\
		0 = & \; \mathcal{D}^{\hat{\pi},\hat{c}} 
		f^{\vert \tau}(t,x,z),   \\
		& V(T,x) =  \log(x),\quad \\
		& \hspace{-0.7cm} f(T,x,z,\tau) =  \dfrac{1}{\rho} \log\rBrackets{e^{-\delta (T-\tau)}  x^{\rho} + z}, \\
	\end{split}
\end{equation}
and the probabilistic representation of $V$ and $f$ is
\begin{align*}
	V(t,x) &= \EVtx\sBrackets{\frac{1}{\rho}\log\rBrackets{\int \limits_{t}^{T}e^{-\delta (s-t)}\left(\hat{c}(s)\right)^{\rho}\,ds + e^{-\delta (T-t) }\rBrackets{X^{\hat{\pi}, \hat{c}}_T}^{\rho}}},\\
	f(t,x,z,\tau) &= \EVtx\sBrackets{\frac{1}{\rho}\log\rBrackets{\int \limits_{t}^{T}e^{-\delta (s-\tau)}\left(\hat{c}(s)\right)^{\rho}\,ds + e^{-\delta (T-\tau) }\rBrackets{X^{\hat{\pi}, \hat{c}}_T}^{\rho} + z}}.
\end{align*}

As in Corollary \ref{cor:CT_KM_with_cur_t_dependence_EHJB_system_infinite}, we can recast the above system in an alternative form. 
\begin{corollary} \label{cor:CT_KM_with_cur_t_dependence_EHJB_system_infinite_alpha=1}
	The extended HJB system for a decision maker with reward functional \eqref{eq:KM_CT_with_cur_t_dependence_reward_functional_alpha=1} can be rewritten as
	\begin{equation}
		\hspace{-0.5cm}
		\begin{split}
			\displaystyle
			0 = & \; \sup_{(\pi,c)\in \mathbb{R}^{2}}\bigg\{\partial_{t}V(t,x) + \partial_x V(t, x) (x(r + \pi \lambda) - c) + \dfrac{1}{2}\partial_{xx}V(t, x) \sigma^2 \pi^2 x^2 \bigg.\\
			& \hspace{2cm} \bigg. + \dfrac{1}{\rho}\widetilde{V}^{(1)}(t,x) c^\rho -\dfrac{\delta }{\rho} \bigg\}, \label{eq:KM_CT_PDE_for_V_alpha=1}\\
			0 = & \; \partial_{t}\widetilde{V}^{(k)}(t,x) + \partial_x\widetilde{V}^{(k)}(t, x) (x(r + \hat{\pi} \lambda) - \hat{c}) +  \dfrac{1}{2}\partial_{xx}\widetilde{V}^{(k)}(t, x) \sigma^2 \hat{\pi}^2 x^2  \\
			& \qquad - \rBrackets{-1}^{k}k\delta\widetilde{V}^{(k)}(t,x) + \rBrackets{-1}^{k}k\widetilde{V}^{(k+1)}(t,x) \hat{c}^{\,\rho} ,  \quad k \in \mathbb{N}, \\
			& V(T,x) = \log \left(x\right),\\
			& \hspace{-0.4cm}\widetilde{V}^{(k)}(T,x) = x^{ -k\rho}, \quad k \in \mathbb{N}. \\
		\end{split}
	\end{equation}
	In addition, the probabilistic representation of $\widetilde{V}^{(k)}, k \in \mathbb{N}$,  is given by
	\begin{align*}
		\widetilde{V}^{(k)}(t,x) &= \EVtx\sBrackets{\rBrackets{\int \limits_{t}^{T}e^{-\delta (s-t) }\left(\hat{c}(s)\right)^{\rho}\,ds + e^{-\delta (T-t)}\rBrackets{X^{\hat{\pi}, \hat{c}}_T}^{\rho}}^{- k}}.
	\end{align*}
\end{corollary}

\begin{proof}
	We carry out the proof via similar arguments as the proof of Corollary \ref{cor:CT_KM_with_cur_t_dependence_EHJB_system_infinite} in Appendix \ref{subsec:ProofCorollaryReformulationHJBSystemCRRA}. We start by computing the terms $\partial_{\tau} f(t, x, 0, t)$ and $\partial_{z} f(t, x, 0, t)$ in \eqref{eq:EHJBsystem_CRRA_unitaryriskaversion}:
	\begin{align*}
		\partial_{\tau} f(t, x, 0, t) \\
		& \hspace{-1.5cm} \stackrel{}{=} \rBrackets{\partial_\tau f(t,x,0,\tau)}_{\tau = t} \\
		&\hspace{-1.5cm}\stackrel{}{=} \rBrackets{\partial_\tau \EVtx \sBrackets{ \frac{1}{\rho}\log\rBrackets{\int \limits_{t}^{T} e^{-\delta \rBrackets{s - \tau}} \left(\hat{c}(s)\right)^{\rho}\,ds + e^{-\delta \rBrackets{T - \tau}}\rBrackets{X^{\hat{\pi}, \hat{c}}_T}^{\rho}  } } }_{\tau = t}\\
		& \hspace{-1.5cm}\stackrel{}{=}  \frac{\delta }{\rho}, \\
		\partial_{z} f(t, x, 0, t) \\
		& \hspace{-1.5cm}\stackrel{}{=} \rBrackets{\partial_{z}  \EVtx\sBrackets{ \frac{1}{\rho}\log\rBrackets{\int \limits_{t}^{T} e^{-\delta \rBrackets{s - t}} \left(\hat{c}(s)\right)^{\rho} \,ds + e^{-\delta \rBrackets{T - t}}\rBrackets{X^{\hat{\pi}, \hat{c}}_T}^{\rho} + z } }}_{z = 0}\\
		& \hspace{-1.5cm} =\EVtx\sBrackets{\frac{1}{\rho}\rBrackets{\int \limits_{t}^{T} e^{-\delta \rBrackets{s - t}} \left(\hat{c}(s)\right)^{\rho} \,ds + e^{-\delta \rBrackets{T - t}}\rBrackets{X^{\hat{\pi}, \hat{c}}_T}^{\rho} }^{-1}},
	\end{align*}
	from which we define $\widetilde{V}^{(1)}(t,x) := \rho \partial_{z} f(t, x, 0, t).$ The first PDE in \eqref{eq:EHJBsystem_CRRA_unitaryriskaversion} then becomes
	\begin{equation}
		0 = \sup_{(\pi,c)\in \mathbb{R}^{2}}\left\{\mathcal{D}^{\pi,c}V(t,x) + \dfrac{1}{\rho}\widetilde{V}^{(1)}(t,x) c^{\rho} -  \frac{\delta}{\rho} \right\},
	\end{equation}
	with terminal condition for $\widetilde{V}^{(1)}(t,x)$ given by $\widetilde{V}^{(1)}(T,x) =  x^{-1}.$ 
	
	Now let
	\begin{equation*}
		f^{(1)}(t,x,z,\tau) = \EVtx\sBrackets{\frac{1}{\rho}\rBrackets{\int \limits_{t}^{T} e^{-\delta \rBrackets{s - \tau}} \left(\hat{c}(s)\right)^{\rho}\,ds + e^{-\delta \rBrackets{T - \tau}}\rBrackets{X^{\hat{\pi}, \hat{c}}_T}^{\rho} + z }^{-1}}.
	\end{equation*}
	As in Lemma \ref{lem:CT_KM_with_tau_recursion_for_f}, $f^{(1)}(t,x,z,\tau)$ satisfies the following PDE:
	\begin{align*}
		0 & = \mathcal{D}^{\hat{\pi},\hat{c}}f^{(1) \vert \tau }(t,x,z) \\
		& =\partial_{t} f^{(1)}(t,x,z,\tau) + \partial_{x}f^{(1)}(t,x,z,\tau)\mu(t,x, (\hat{\pi}(t,x),\hat{c}(t,x))) \\
		& \qquad + \frac{1}{2}\partial_{xx}f^{(1)}(t,x,z,\tau) \rBrackets{\sigma(t, x, (\hat{\pi}(t,x),\hat{c}(t,x))}^2 \\
		& \qquad + \partial_{z}f^{(1)}(t,x,z,\tau)(\hat{c}(t,x))^{\, \rho}.
	\end{align*}
	Inserting $z = 0$ and $\tau = t$ above, and using the definition of $\widetilde{V}^{(1)}(t,x)$, we obtain
	\begin{align*}
		0 & \stackrel{}{=} \partial_{t}\widetilde{V}^{(1)}(t,x) + \partial_{x}\widetilde{V}^{(1)}(t,x)\mu(t,x, (\hat{\pi}(t,x),\hat{c}(t,x)))  \\
		& \qquad + \frac{1}{2}\partial_{xx}\widetilde{V}^{(1)}(t,x) \rBrackets{\sigma(t, x, (\hat{\pi}(t,x),\hat{c}(t,x))}^2 - \rho \partial_{\tau} f(t, x, 0, t) \\
		& \qquad + \rho \partial_{z}f^{(1)}(t,x,0, t)(\hat{c}(t,x))^{\, \rho}\\
		& \stackrel{}{=} \mathcal{D}^{\hat{\pi},\hat{c}} \, \widetilde{V}^{(1)}(t,x) - \rho \partial_{\tau}f^{(1)}(t, x, 0, t) + \rho \partial_{z}f^{(1)}(t,x,0, t)(\hat{c}(t,x))^{\, \rho}.
	\end{align*}
	Then, observe that $
	\partial_{\tau}f^{(1)}(t,x,0, t)  = -\frac{\delta}{\rho}\widetilde{V}^{(1)}(t,x), $
	and
	\begin{align*}
		\partial_z f^{(1)}(t,x,0, t) \\
		& \hspace{-1.5cm} =  \left(\partial_z\EVtx\sBrackets{\frac{1}{\rho}\rBrackets{\int \limits_{t}^{T} e^{-\delta \rBrackets{s - t}} \left(\hat{c}(s)\right)^{\rho}\,ds + e^{-\delta \rBrackets{T - t}}\rBrackets{X^{\hat{\pi}, \hat{c}}_T}^{\rho} + z}^{-1}} \right)_{z=0}\\
		& \hspace{-1.5cm} = \left( - \EVtx\sBrackets{\frac{1}{\rho}\rBrackets{\int \limits_{t}^{T} e^{-\delta \rBrackets{s - t}} \left(\hat{c}(s)\right)^{\rho}\,ds + e^{-\delta \rBrackets{T - t}}\rBrackets{X^{\hat{\pi}, \hat{c}}_T}^{\rho} +z }^{-2}} \right)_{z=0}, 
	\end{align*}
	from which we define $\widetilde{V}^{(2)}(t,x) := - \rho \partial_z f^{(1)}(t,x,0, t) $.
	Therefore, the PDE for $\widetilde{V}^{(1)}$ reads as follows:
	\begin{equation*}
		0 = \mathcal{D}^{\hat{\pi},\hat{c}}\,\widetilde{V}^{(1)}(t,x) +\delta\widetilde{V}^{(1)}(t,x) - \widetilde{V}^{(2)}(t,x)(\hat{c}(t,x))^{\, \rho}.
	\end{equation*}
	Repeating the same procedure for all $k \geq 2$, we are led to the generic PDE for $\widetilde{V}^{(k)}(t,x)$:
	\begin{align*}
		0 = & \;  \mathcal{D}^{\hat{\pi},\hat{c}}\,\widetilde{V}^{(k)}(t,x) -  \rBrackets{-1}^{k}k\delta\widetilde{V}^{(k)}(t,x) +  \rBrackets{-1}^{k}k\widetilde{V}^{(k+1)}(t,x) (\hat{c}(t,x))^{\, \rho},\\
		& \widetilde{V}^{(k)}(T,x)  =  x^{- k\rho}.
	\end{align*}
	This proves the claim of the corollary. 
\end{proof}

The last passage aims to derive the system of ODEs that characterizes the solution to the problem with unit RRA. From the first-order conditions for the supremum in \eqref{eq:KM_CT_PDE_for_V_alpha=1}, we obtain the candidate equilibrium controls:
\begin{align*}
	\hat{\pi}(t,x) &= -\frac{\partial_x V(t, x) \lambda}{\partial_{xx} V(t, x) x \sigma^2},\\
	\hat{c}(t,x) &= \rBrackets{\frac{\partial_x V(t, x)}{\widetilde{V}^{(1)}(t,x)}}^{\frac{1}{\rho - 1}}.
\end{align*}
Next, we conjecture that the variables $t$ and $x$ can be separated via the following ansatz:
\begin{align*}
	V(t,x) &=  B(t) \log(x) + L(t)\\
	\widetilde{V}^{(k)}(t,x) &=  B^{(k)}(t) x^{- k\rho}, \quad k \in \mathbb{N},
\end{align*}
where $B(t), L(t), B^{(k)}(t)$ are functions to be determined. Using the ansatz, we rewrite the (candidate) equilibrium strategies as
\begin{align*}
	\hat{\pi}(t,x) &= \dfrac{\lambda}{ \sigma^{2} },\\
	\hat{c}(t,x) &= x \left(\dfrac{B(t)}{B^{(1)}(t)}\right)^{\frac{1}{\rho-1}}.
\end{align*}
Finally, going through similar calculations as in the case of $\alpha \neq 1$, we end up with a system of infinitely many ODEs:
\begin{equation*}
\begin{split}
		0 = & \, \partial_{t}B(t), \\
		0  = & \, \partial_{t}L(t) + B(t)  \left(  r + \dfrac{\lambda^{2}}{2\sigma^{2}}  - \left(\dfrac{B(t)}{B^{(1)}(t)}\right)^{\frac{1}{\rho-1}} \right)  + \dfrac{1}{\rho}B^{(1)}(t) \left(\dfrac{B(t)}{B^{(1)}(t)}\right)^{\frac{\rho}{\rho-1}} - \dfrac{\delta}{\rho},\\
		0  = & \, \partial_{t}B^{(k)}(t)  - k\rho B^{(k)}(t) \left(r + \dfrac{\lambda^{2}}{2\sigma^{2}} - \left(\dfrac{B(t)}{B^{(1)}(t)}\right)^{\frac{1}{\rho-1}}   \right) +\dfrac{1}{2}k^{2}\rho^{2} B^{(k)}(t) \dfrac{\lambda^{2}}{\sigma^{2}}    \\
		& \, \qquad - \rBrackets{-1}^{k}k\delta B^{(k)}(t)  +   \rBrackets{-1}^{k}kB^{(k+1)}(t)\, \left(\dfrac{B(t)}{B^{(1)}(t)}\right)^{\frac{\rho}{\rho-1}}  , \quad k \in \mathbb{N},\\
		& B(T) = 1, \\
		& L(T) = 0, \\
		& B^{(k)}(T) = 1, \quad k \in \mathbb{N}.  
\end{split}
\end{equation*}

\end{appendix}

\bibliography{BibFile}  

\end{document}